\documentclass[final,leqno,showlabe]{siamltex}
\usepackage{amsmath}
\usepackage{amssymb}
\usepackage{cases}
\usepackage{enumerate}

\input psfig.sty

\newcommand{\bP}{{\bf P}}

\newcommand{\p}{\partial}

\newcommand{\vep}{\varepsilon}

\newcommand{\be}{\begin{equation}}
\newcommand{\ee}{\end{equation}}
\newcommand{\ba}{\begin{array}}
\newcommand{\ea}{\end{array}}
\newcommand{\bea}{\begin{eqnarray}}
\newcommand{\eea}{\end{eqnarray}}
\newcommand{\beas}{\begin{eqnarray*}}
\newcommand{\eeas}{\end{eqnarray*}}

\newtheorem{remark}{Remark}[section]

 \newcommand{\bx}{{\bf x} }

\renewcommand{\ldots}{\dotsc}

\title{Ground states and dynamics of spin-orbit-coupled Bose-Einstein condensates}
\author{Weizhu Bao\thanks{Department of Mathematics and
Center for Computational Science and Engineering, National
University of Singapore, Singapore 119076 ({\it
matbaowz@nus.edu.sg}, URL: http://www.math.nus.edu.sg/\~{}bao/).}
\and Yongyong Cai\thanks{Beijing Computational Science Research Center, Beijing 100084,
P. R. China; and Department of Mathematics, Purdue University, West Lafayette,
IN 47907, USA ({\it yonyong.cai@gmail.com}).} }

\date{}
\begin{document}

\maketitle

\begin{abstract}
We study analytically and asymptotically as well as numerically ground states and dynamics of
two-component spin-orbit-coupled Bose-Einstein condensates (BECs) modeled by
the coupled Gross-Pitaevskii equations (CGPEs).
In fact, due to the appearance of the spin-orbit (SO) coupling in the two-component BEC with a
Raman coupling,
the ground state structures and dynamical properties become very rich and complicated.
For the ground states, we establish the existence and  non-existence
results under different parameter regimes, and obtain their limiting behaviors and/or structures
with different combinations of the SO and
Raman coupling strengths. For the dynamics, we show that the
motion of the center-of-mass is either non-periodic or with different frequency
to the trapping frequency when the external trapping potential is taken as harmonic
and the initial data is  chosen as a stationary state (e.g. ground state) with
a shift, which is completely different from the case of a two-component BEC without the SO coupling,
and obtain the semiclassical limit of the CGPEs in the linear case via
the Wigner transform method. Efficient and accurate numerical methods
are proposed for computing the ground states and dynamics, especially for the case of box potentials.
Numerical results are reported to demonstrate the efficiency and accuracy of the numerical
methods and show the rich phenomenon in the SO-coupled BECs.
\end{abstract}



\pagestyle{myheadings} \markboth{Weizhu Bao and Yongyong Cai}
{Ground states and dynamics of spin-orbit-coupled BEC}

\section{Introduction}
\label{s1} \setcounter{equation}{0}
 Spin-orbit  (SO) coupling  is the interaction
between the spin and  motion of a particle, and is crucial for understanding
many physical phenomenon, such as quantum Hall effects \cite{Xiao} and topological insulators \cite{Hasan}.
However, SO coupling observation in solid state matters is inaccurate due to
the disorder and impurities of the system.  Since
the first experimental realization of Bose-Einstein condensation (BEC)
in 1995 \cite{Anderson,Davis}, degenerate quantum gas has become a perfect candidate for studying
quantum many-body phenomenon in condensed matter physics. Such a system of quantum gas can be controlled
with high precision in experiments.   Very recently, in a pioneer work~\cite{Lin2011}, Lin et al. have
created a spin-orbit-coupled BEC with two spin states of
$^{85}$Rb: $\left|\uparrow\rangle\right.=|F=1,\,m_f=0\rangle$ and
$\left|\downarrow\rangle\right.=|F=1,\,m_f=-1\rangle$.  Due to this remarkable experimental
progress and its potential applications,  SO coupling in cold atoms
has received broad interests in atomic physics community and condensed matter physics community  \cite{Galitski,Hamner}.

At temperatures $T$ much smaller than the critical temperature
$T_c$, following the mean field theory \cite{Lin2011,PitaevskiiStringari,ZhangJ}, a SO-coupled BEC is well described by
the macroscopic wave function $\Psi:=\Psi(\bx,t)=(\psi_{1}(\bx,t),\psi_{2}(\bx,t))^T:=(\psi_1,\psi_2)^T$
whose evolution is governed by the coupled Gross-Pitaevskii equations (CGPEs) in  three dimensions (3D)
\be \label{cgpe1}
\begin{split}
&i\hbar\partial_t \psi_{1}=\left[-\frac{\hbar^2}{2m}\nabla^2
+\tilde{V}_1(\bx)+\frac{i\hbar^2 \tilde{k}_0}{m}\p_x+\frac{\hbar\tilde{\delta}}{2}
+\sum_{l=1}^2 \tilde{g}_{1l}|\psi_{l}|^2\right]\psi_{1}+\frac{\hbar\tilde{\Omega}}{2}
\psi_{2}, \\
&i\hbar\partial _t \psi_{2}=\left[-\frac{\hbar^2}{2m}\nabla^2
+\tilde{V}_2(\bx)-\frac{i\hbar^2\tilde{k}_0}{m} \p_x-\frac{\hbar\tilde{\delta}}{2}+
\sum_{l=1}^2 \tilde{g}_{2l}|\psi_{l}|^2\right]\psi_{2}
+\frac{\hbar\tilde{\Omega}}{2}
\psi_{1}.\end{split} \ee
Here, $t$ is time,
$\bx=(x,y,z)^T\in \mathbb R^3$   is the Cartesian coordinate
vector, $\hbar$ is the Planck constant, $m$ is the mass of particle,
$\tilde{\delta}$ is the detuning constant for Raman transition,
$\tilde{k}_0$
is the wave number of Raman lasers representing the SO coupling strength,
$\tilde{\Omega}$ is the effective Rabi frequency describing the strength of
Raman coupling (i.e. an internal atomic Josephson junction), and
$\tilde{g}_{jl}=\frac{4\pi \hbar^2 a_{jl}}{m}$ ($j,l=1,2$) are
interaction constants with $a_{jl}=a_{lj}$ ($j,l=1,2$) being the $s$-wave scattering
lengths between the $j$th and $l$th component (positive for
repulsive interaction and negative for attractive interaction).
$\tilde{V}_{1}(\bx)$ and $\tilde{V}_{2}(\bx)$  are given
real-valued external trapping potentials whose profiles depend on
different applications and the setups in experiments \cite{Lin2011,Hamner}.
In typical current experiments, the following harmonic potentials are commonly used
\cite{Lin2011,Hamner,Li2012,Li2013}
\begin{equation}\label{trap}
\tilde{V}_j(\bx)=\frac{m}{2}\left[\omega_x^2x^2+\omega_y^2y^2+\omega_z^2(z-\tilde{z}_j)^2\right],
 \qquad j=1,2,\quad\bx=(x,y,z)^T\in {\mathbb R}^3,
\end{equation}
where $\omega_x>0$, $\omega_y>0$ and $\omega_z>0$ are trapping frequencies in $x$-, $y$- and $z$-direction,
respectively, and $\tilde{z}_1,\tilde{z}_2\in {\mathbb R}$ are two given constants.
The wave function $\Psi$ is  normalized as
\be \label{norm}
\|\Psi\|^2:=\|\Psi(\cdot,t)\|_2^2=\int_{{\mathbb R}^3}
\left[|\psi_1(\bx,t)|^2+|\psi_2(\bx,t)|^2\right]\,d\bx=N,\ee
where $N$ is the total number of particles in the SO-coupled BEC.

In order to nondimensionalize the CGPEs (\ref{cgpe1}) with (\ref{trap}), we introduce \cite{Bao2013,Bao2014}
\begin{equation}\label{scal}
\begin{split}
\tilde{t}=\frac{t}{t_s},\qquad \tilde{\bx}=\frac{\bx}{x_s},\qquad
 \tilde{\psi}_j(\tilde{\bx},\tilde{t})=\frac{x_s^{3/2}}{N^{1/2}}\psi_j(\bx,t), \quad j=1,2,
\end{split}
\end{equation}
where $t_s=\frac{1}{\omega_0}$ and $x_s=\sqrt{\frac{\hbar}{m\omega_0}}$ with
$\omega_0=\min\{\omega_x,\omega_y,\omega_z\}$
are the scaling parameters of dimensionless time and length units, respectively.
Plugging (\ref{scal}) into  \eqref{cgpe1}, multiplying by
$\frac{t_s^2}{m(x_sN)^{1/2}}$, and  then removing all $\tilde{}$,  we obtain  the
following dimensionless CGPEs in 3D for a SO-coupled BEC
\be\label{eq:cgpe13d}
\begin{split} &i\partial_t \psi_{1}=\left[-\frac{1}{2}\nabla^2
+V_1(\bx)+ik_0\p_x+\frac{\delta}{2}
+\left(g_{11}|\psi_{1}|^2+g_{12}|\psi_{2}|^2\right)\right]\psi_{1}+\frac{\Omega}{2}
\psi_{2}, \\
&i\partial _t \psi_{2}=\left[-\frac{1}{2}\nabla^2
+V_2(\bx)-ik_0 \p_x-\frac{\delta}{2}+\left(g_{21}|\psi_{1}|^2+g_{22}|\psi_{2}|^2\right)\right]\psi_{2}
+\frac{\Omega}{2}
\psi_{1},\end{split} \ee
where $k_0=\frac{\tilde{k}_0x_s}{2}$, $\delta=\frac{\tilde{\delta}}{\omega_0}$, $\Omega=\frac{\tilde{\Omega}}{\omega_0}$,
$g_{11}=\frac{4\pi Na_{11}}{x_s}$, $g_{12}=g_{21}=\frac{4\pi Na_{12}}{x_s}$,
$g_{22}=\frac{4\pi Na_{22}}{x_s}$, $\gamma_x=\frac{\omega_x}{\omega_0}$,
$\gamma_y=\frac{\omega_y}{\omega_0}$ and $\gamma_z=\frac{\omega_z}{\omega_0}$,
and the dimensionless trapping potentials are
\begin{equation} \label{trap72}
V_j(\bx)=\frac{1}{2}\left(\gamma_x^2x^2+\gamma_y^2y^2+\gamma_z^2(z-z_j)^2\right),\quad \bx\in{\mathbb R}^3,\qquad
j=1,2,
\end{equation}
with $z_1=\frac{\tilde{z}_1}{x_s}$ and $z_2=\frac{\tilde{z}_2}{x_s}$.

When the trapping potentials in (\ref{trap72}) are strongly anisotropic, similar to
the dimension reduction of the GPE for a BEC \cite{Bao2013,Bao2014,Ben,PitaevskiiStringari},
the CGPEs (\ref{eq:cgpe13d}) in 3D can be formally reduced
to two dimensions (2D) or one dimension (1D) when the BEC is disk-shaped or cigar-shaped, respectively.
For simplicity of notations, we assume $z_1=z_2=0$ in (\ref{trap72}).
When $\gamma_z\gg \gamma_x$ and $\gamma_z\gg\gamma_y$, i.e. a disk-shaped condensate,
by taking the  ansatz \cite{Ben,Bao2013}
\be\label{dimred}
\psi_j(\bx,t)=\psi_{j}^{2D}(x,y,t)e^{-i\gamma_zt/2}\gamma_z^{-1/4}
w(\gamma_z^{1/2}z), \quad \bx=(x,y,z)^T\in{\mathbb R}^3,\quad j=1,2,
\ee
with $w(z)=\pi^{-1/4}e^{-z^2/2}$, multiplying both sides of \eqref{eq:cgpe13d} by $w(\gamma_z^{1/2}z)$ and integrating
over $z\in{\mathbb R}$, we can formally reduce the 3D CGPEs \eqref{eq:cgpe13d} into 2D as \cite{Bao2013,Bao2014}
\begin{equation}\label{cgpe71}
\begin{split}
&i\partial_t \psi_{1}^{2D}=\left[-\frac{1}{2}\nabla^2
+V_1^{2D}(x,y)+ik_0\p_x+\frac{\delta}{2}
+\sum_{l=1}^2g_{1l}^{2D}|\psi_{l}^{2D}|^2\right]\psi_{1}^{2D}+\frac{\Omega}{2}
\psi_{2}^{2D}, \\
&i\partial _t \psi_{2}^{2D}=\left[-\frac{1}{2}\nabla^2
+V_2^{2D}(x,y)-ik_0 \p_x-\frac{\delta}{2}+\sum_{l=1}^2g_{2l}^{2D}|\psi_{l}^{2D}|^2\right]\psi_{2}^{2D}
+\frac{\Omega}{2}
\psi_{1}^{2D},\end{split} \end{equation}
where $g_{jl}^{2D}\approx \frac{\sqrt{\gamma_z}}{\sqrt{2\pi}}g_{jl}$ ($j,l=1,2$) and $V_1^{2D}(x,y)=V_2^{2D}(x,y)=\frac{1}{2}(\gamma_x^2x^2+\gamma_y^2y^2)$.
Similarly,  when $\gamma_z\gg\gamma_x$ and $\gamma_y\gg \gamma_x$, i.e. a cigar-shaped condensate,
we can formally reduce the 3D CGPEs \eqref{eq:cgpe13d} into 1D as \cite{Bao2013,Bao2014,Ben,PitaevskiiStringari}
\begin{equation}\label{cgpe72}
\begin{split} &i\partial_t \psi_{1}^{1D}=\left[-\frac{1}{2}\nabla^2
+V_1^{1D}(x)+ik_0\p_x+\frac{\delta}{2}
+\sum_{l=1}^2g_{1l}^{1D}|\psi_{l}^{1D}|^2\right]\psi_{1}^{1D}+\frac{\Omega}{2}
\psi_{2}^{1D}, \\
&i\partial _t \psi_{2}^{1D}=\left[-\frac{1}{2}\nabla^2
+V_2^{1D}(x)-ik_0 \p_x-\frac{\delta}{2}+\sum_{l=1}^2g_{2l}^{1D}|\psi_{l}^{1D}|^2\right]\psi_{2}^{1D}
+\frac{\Omega}{2}
\psi_{1}^{1D},\end{split} \end{equation}
where $g_{jl}^{1D}\approx\frac{\sqrt{\gamma_y\gamma_z}}{2\pi}g_{jl}$ ($j,l=1,2$) and $V_1^{1D}(x)=V_2^{1D}(x)=\frac{1}{2}\gamma_x^2x^2$.

In fact, the CGPEs (\ref{eq:cgpe13d}) in 3D, (\ref{cgpe71}) in 2D and (\ref{cgpe72}) in 1D can be written
in a unified form in $d$-dimensions ($d=1,2,3$) for $\bx\in \mathbb R^d$ with $\bx=x\in\mathbb R$,
$\psi_1=\psi_1^{1D}$, $\psi_2=\psi_2^{1D}$ and $\beta_{jl}=\frac{\sqrt{\gamma_y\gamma_z}}{2\pi}g_{jl}$ for $d=1$;
$\bx=(x,y)^T\in\mathbb R^2$,
 $\psi_1=\psi_1^{2D}$, $\psi_2=\psi_2^{2D}$ and $\beta_{jl}=\frac{\sqrt{\gamma_z}}{\sqrt{2\pi}}g_{jl}$ for $d=2$; and
$\bx=(x,y,z)^T\in\mathbb R^3$ and $\beta_{jl}=g_{jl}$ ($j,l=1,2$) for $d=3$ as
\be\label{eq:cgpe1}
\begin{split} &i\partial_t \psi_{1}=\left[-\frac{1}{2}\nabla^2
+V_1(\bx)+ik_0\p_x+\frac{\delta}{2}
+(\beta_{11}|\psi_{1}|^2+\beta_{12}|\psi_{2}|^2)\right]\psi_{1}+\frac{\Omega}{2}
\psi_{2}, \\
&i\partial _t \psi_{2}=\left[-\frac{1}{2}\nabla^2
+V_2(\bx)-ik_0 \p_x-\frac{\delta}{2}+(\beta_{21}|\psi_{1}|^2+\beta_{22}|\psi_{2}|^2)\right]\psi_{2}
+\frac{\Omega}{2}
\psi_{1},\end{split} \ee
where
\be\label{eq:potential:ho}
V_1(\bx)=V_2(\bx)=\begin{cases}
\frac12(\gamma_x^2x^2+\gamma_y^2y^2+\gamma_z^2z^2),& d=3,\\
\frac12(\gamma_x^2x^2+\gamma_y^2y^2),& d=2,\\
\frac12\gamma_x^2x^2, &d=1,
\end{cases} \qquad \bx\in{\mathbb R}^d.
\ee

For other potentials such as box potential, optical lattice potential and double-well potential,
we refer to \cite{Bao2013,Lin2011,Hamner,Li2012,Li2013,PitaevskiiStringari} and references therein.
Thus, in the subsequent discussion, we will treat
the external potentials $V_1(\bx)$ and  $V_2(\bx)$ in (\ref{eq:cgpe1})
as two general real-valued functions and $\beta_{jl}$ ($j,l=1,2$) satisfying $\beta_{12}=\beta_{21}$ as
arbitrary real constants.  In addition,
without loss of generality, we assume $V_1(\bx)\ge0$ and $V_2(\bx)\ge0$ for $\bx\in{\mathbb R}^d$
in the rest of this paper.
The dimensionless CGPEs (\ref{eq:cgpe1}) conserve the total mass or
normalization, i.e.
 \be\label{eq:mass1}
N(t):=\|\Psi(\cdot,t)\|^2=\int_{{\mathbb R}^d}[|\psi_1(\bx,t)|^2+|\psi_2(\bx,t)|^2]d\bx
\equiv \|\Psi(\cdot,0)\|^2=1, \quad t\ge0, \ee
and the energy per particle
\begin{eqnarray} \label{energy91}
E(\Psi)&=&\int_{{\mathbb R}^d}\biggl[
\sum\limits_{j=1}^2\left(\frac12|\nabla\psi_{j}|^2+V_j(\bx)|\psi_j|^2\right)+\frac{\delta}{2}\left(
|\psi_1|^2-|\psi_2|^2\right)+\Omega\;\text{Re}
(\psi_1\overline{\psi}_2)\nonumber\\
&&\qquad+ik_0\left(\overline{\psi}_1\p_x\psi_1-\overline{\psi}_2\p_x\psi_2\right) +\frac{\beta_{11}}{2}|\psi_1|^4+\frac{\beta_{22}}{2}|\psi_2|^4+\beta_{12}|\psi_1|^2|\psi_2|^2
\biggl]d\bx, \label{eq:energy}
\end{eqnarray}
where $\overline{f}$ and Re$(f)$ denote the conjugate and real part of a
function $f$, respectively. In addition, if  $\Omega=0$ in (\ref{eq:cgpe1}), the mass of each
component is also conserved, i.e. \be\label{eq:Njtt1}
N_j(t):=\|\psi_j(\bx,t)\|^2=\int_{{\mathbb
R}^d} |\psi_j(\bx,t)|^2\,d\bx\equiv \|\psi_j(\bx,0)\|^2, \quad t\ge0, \quad j=1,2.\ee

Finally, by introducing the following change of variables
\be\label{changv}
\psi_1(\bx,t)=\tilde{\psi}_1(\bx,t)
 e^{i(\omega t+k_0x)}, \quad \psi_2(\bx,t)=\tilde{\psi}_2(\bx,t)
 e^{i(\omega t-k_0x)},\qquad \bx\in{\mathbb R}^d,
 \ee
with $\omega=\frac{-k_0^2}{2}$
in the CGPEs (\ref{eq:cgpe1}), we obtain for $\bx\in{\mathbb R}^d$ and $t>0$
\begin{equation}\label{eq:cgpe199:sec9}
\begin{split}
&i\partial_t \tilde{\psi}_1=\left[-\frac{1}{2}\nabla^2
+V_1(\bx)+\frac{\delta}{2}
+\beta_{11}|\tilde{\psi}_1|^2+\beta_{12}|\tilde{\psi}_2|^2\right]\tilde{\psi}_1+\frac{\Omega}{2}
e^{-i2k_0x}\tilde{\psi}_2, \\
&i\partial _t \tilde{\psi}_2=\left[-\frac{1}{2}\nabla^2
+V_2(\bx)-\frac{\delta}{2}+\beta_{21}|\tilde{\psi}_1|^2+\beta_{22}|\tilde{\psi}_2|^2\right]
\tilde{\psi}_2+\frac{\Omega}{2}e^{i2k_0x}
\tilde{\psi}_1.
\end{split}
\end{equation}
For any $\Omega\in {\mathbb R}$, the above CGPEs (\ref{eq:cgpe199:sec9}) conserve the normalization (\ref{eq:mass1}),
i.e. $N(t)=\|\tilde{\Psi}(\cdot,t)\|^2\equiv \|\tilde{\Psi}(\bx,0)\|^2=1$ for $t\ge0$
with $\tilde{\Psi}=(\tilde{\psi}_1,\tilde{\psi}_2)^T$ and the energy per particle
\begin{eqnarray} \tilde{E}(\tilde{\Psi})&=&\int_{{\mathbb R}^d}\biggl[
\sum\limits_{j=1}^2\left(\frac12|\nabla\tilde{\psi}_{j}|^2+V_j(\bx)|\tilde{\psi}_j|^2\right)+\frac{\delta}{2}\left(
|\tilde{\psi}_1|^2-|\tilde{\psi}_2|^2\right)+\Omega\;\text{Re}
(e^{i2k_0x}\tilde{\psi}_1\overline{\tilde{\psi}}_2)\nonumber\\
&&\qquad +\frac{\beta_{11}}{2}|\tilde{\psi}_1|^4+\frac{\beta_{22}}{2}|\tilde{\psi}_2|^4+\beta_{12}|\tilde{\psi}_1|^2
|\tilde{\psi}_2|^2
\biggl]d\bx. \label{eq:energy1}
\end{eqnarray}

In fact, different proposals resulting in different theoretical models
have been proposed in the literatures
for realizing SO-coupled BECs in experiments
\cite{Wang,Lin2011,Hamner,ZhangJ,Deng,Hu}.
Based on these proposed mean field models including the CGPEs \eqref{eq:cgpe1},
ground state structures and dynamical properties of SO-coupled BECs have been
theoretically studied and predicted in the literatures, including phase transition \cite{Hu},
spin vortex structure \cite{Deng}, motion of the center-of-mass \cite{Zhang},
Bogoliubov excitation \cite{Zhu}, etc.
To the best of our knowledge, only the model described by the CGPEs \eqref{eq:cgpe1}
 has been realized experimentally for a SO-coupled BEC \cite{Lin2011,Hamner,ZhangJ}.
Other models have not been realized in experiments yet. Thus we will present
our results on ground states and dynamics of SO-coupled BECs based on
the CGPEs (\ref{eq:cgpe1}). We remark that our methods and results
are still valid for other theoretical models for  SO-coupled BECs in the literatures
\cite{Wang,Hamner,ZhangJ,Deng,Hu}.

For the CGPEs (\ref{eq:cgpe1}), when $k_0=0$, i.e.,
a two-component BEC without SO coupling and without/with Raman coupling corresponding
to $\Omega=0$/$\Omega\ne0$, ground state structures and dynamical properties have been
studied theoretically in
the literature \cite{Bao1,Chang,LS2,Bao2009,Liu}.
When the SO coupling is taken into consideration, i.e. $k_0\neq0$, when $\Omega=0$,
it can be easily removed from the CGPEs (\ref{eq:cgpe1}) via (\ref{changv}) and
thus the SO coupling has no essential effect to the system. Therefore
in order to observe the effect of the SO coupling, $\Omega$ must be chosen nonzero.
To the best of our knowledge, there exist very few mathematical results
to the CGPEs (\ref{eq:cgpe1}) when $k_0\ne0$ and $\Omega\ne0$ in the literature.
The main aim of this paper is to mathematically study the existence of ground states and their
 structures as well as dynamical properties of SO-coupled BECs based on
 the CGPEs \eqref{eq:cgpe1} and propose efficient and
accurate methods for numerically simulating ground states and dynamics.

The paper is organized as follows. In section 2, we
establish existence and  non-existence
results of ground states under different parameter regimes,
and obtain their limiting behaviors and/or structures
with different combinations of the SO and
Raman coupling strengths.
In section 3, we present efficient and accurate numerical methods
for computing ground states and dynamics of SO-coupled BECs and report
ground states for different parameter regimes.
In section 4, we derive dynamical properties on the motion of the center-of-mass,
compare them with numerical results,
and obtain the semiclassical limit of the CGPEs  in the linear case via
the Wigner transform method.
Finally, some conclusions are drawn in section 5. Throughout the paper,
we adopt standard notations of the Sobolev spaces.

\section{Ground states}
\label{s2} \setcounter{equation}{0}
The ground state $\Phi_g:=\Phi_g(\bx)=(\phi_1^g(\bx),\phi_2^g(\bx))^T$ of
a two-component SO-coupled BEC based on (\ref{eq:cgpe1})
is defined as the minimizer of the energy functional (\ref{energy91}) under the
constraint (\ref{eq:mass1}), i.e.
\begin{quote}
  Find $\Phi_g \in S$, such that
\end{quote}
  \begin{equation}\label{eq:minimize}
    E_g := E\left(\Phi_g\right) = \min_{\Phi \in S}
    E\left(\Phi\right),
  \end{equation}
where $S$ is  defined as \be\label{eq:nonconset}
S:=\left\{\Phi=(\phi_1,\phi_2)^T\in H^1(\mathbb R^d)^2 \ | \ \|\Phi\|^2=\int_{{\mathbb
R}^d}\left(|\phi_1(\bx)|^2+|\phi_2(\bx)|^2\right)d\bx=1\right\}.\ee
Since $S$ is a nonconvex set,  the problem (\ref{eq:minimize}) is a nonconvex minimization
problem.  In addition, the ground state $\Phi_g$ is a solution to the following nonlinear
eigenvalue problem, i.e. Euler-Lagrange equation of the problem (\ref{eq:minimize})
\be\label{eq:e-l:1} \begin{split}
&\mu \phi_{1}=\left[-\frac{1}{2}\nabla^2
+V_1(\bx)+ik_0\p_x+\frac{\delta}{2}
+(\beta_{11}|\phi_{1}|^2+\beta_{12}|\phi_{2}|^2)\right]\phi_{1}+\frac{\Omega}{2}
\phi_{2}, \\
&\mu \phi_{2}=\left[-\frac{1}{2}\nabla^2
+V_2(\bx)-ik_0\p_x-\frac{\delta}{2}+(\beta_{12}|\phi_{1}|^2+\beta_{22}|\phi_{2}|^2)\right]\phi_{2}
+\frac{\Omega}{2}
\phi_{1},\end{split} \ee
under the normalization constraint $\Phi\in S$. For an eigenfunction
 $\Phi=(\phi_1,\phi_2)^T$ of (\ref{eq:e-l:1}), its corresponding eigenvalue (or chemical potential in the physics
literature)  $\mu:=\mu(\Phi)=\mu(\phi_1,\phi_2)$ can be computed as
\be
\mu=E(\Phi)+\int_{\mathbb R^d}\left(\frac{\beta_{11}}{2}|\phi_1|^4+\frac{\beta_{22}}{2}|\phi_2|^4
+\beta_{12}|\phi_1|^2|\phi_2|^2\right)\,d\bx.
\ee

Similarly, the ground state $\tilde{\Phi}_g=(\tilde{\phi}_1^g, \tilde{\phi}_2^g)^T\in S$
of \eqref{eq:cgpe199:sec9} is defined as:
 \begin{quote}
  Find $\tilde{\Phi}_g \in S$, such that
\end{quote}
  \begin{equation}\label{eq:minimize1}
    \tilde{E}_g := \tilde{E}\left(\tilde{\Phi}_g\right) = \min_{\tilde\Phi \in S}
    \tilde{E}\left(\tilde\Phi\right).
  \end{equation}
 We notice that the ground state $\Phi_g=(\phi_1^g, \phi_2^g)^T$ given by \eqref{eq:minimize} has one-to-one correspondence
 with the ground state $\widetilde{\Phi}_g=(\tilde{\phi}_1^g, \tilde{\phi}_2^g)^T$ given by \eqref{eq:minimize1}, through the following relation
 \be
\Phi_g=(\phi_1^g, \phi_2^g)^T=(e^{ik_0x}\tilde{\phi}_1^g,
 e^{-ik_0x}\tilde{\phi}_2^g)^T \ \Longleftrightarrow \ \tilde\Phi_g=(\tilde\phi_1^g,\tilde\phi_2^g)^T=
 (e^{-ik_0x}\phi_1^g,e^{ik_0x}\phi_2^g)^T.
 \ee
In the sequel, the $\tilde{}$ acting on $\Phi=(\phi_1,\phi_2)^T$ always means that
\be\label{eq:equiv}
\tilde{\Phi}=(\tilde{\phi}_1,\tilde{\phi}_2)^T=(e^{-ik_0x}\phi_1,e^{ik_0x}\phi_2)^T
\ \Longleftrightarrow \  \Phi=(\phi_1,\phi_2)^T=(e^{ik_0x}\tilde\phi_1,e^{-ik_0x}\tilde\phi_2)^T,
\ee
and the following equality holds
\be\label{eq:reform}
E(\Phi)=\tilde{E}(\tilde{\Phi})-\frac{k_0^2}{2}\|\Phi\|^2=\tilde{E}(\tilde{\Phi})-\frac{k_0^2}{2}\|\tilde{\Phi}\|^2.
\ee
In particular
\be
E(\Phi)=\tilde{E}(\tilde{\Phi})-\frac{k_0^2}{2}, \qquad \Phi\in S.
\ee
When $k_0=0$, the existence and uniqueness as well as non-existence results of the ground state
of the problem (\ref{eq:minimize}) have been studied in \cite{Bao2009}.
Hereafter, we  assume $k_0\neq0$.

\subsection{Existence and uniqueness}


In 2D, i.e. $d=2$, let $C_b$ be the
best constant in the following inequality \cite{Weinstein}
\be\label{bestc}
C_b:=\inf_{0\ne f\in H^1({\mathbb R}^2)} \frac{\|\nabla
f\|_{L^2(\mathbb R^2)}^2\|f\|_{L^2(\mathbb R^2)}^2}{\|f\|_{L^4(\mathbb R^2)}^4}.
\ee

Define the function $I(\bx)$ as
\be\label{eq:Ibx}
I(\bx)=\left(V_1(\bx)-V_2(\bx)+\delta\right)^2 +(\beta_{11}-\beta_{12})^2+(\beta_{12}-\beta_{22})^2,
\ee
where $I(\bx)\equiv0$ means that the SO coupled BEC with $k_0=\Omega=0$ is essentially one component;
denote the interaction coefficient matrix
\be\label{beta97}
 A=\begin{pmatrix}\beta_{11}&\beta_{12}\\
\beta_{21}&\beta_{22}\end{pmatrix}= A^T,
\ee
and $A$ is said to be nonnegative if $\beta_{jl}\ge0$ ($j,l=1,2$);

Introduce the function space
\begin{equation*}
X=\left\{(\phi_1,\phi_2)^T\in H^1(\mathbb R^d)\times H^1(\mathbb R^d)\left| \int_{\mathbb R^d}\left(V_1(\bx)|\phi_1(\bx)|^2+V_2(\bx)|\phi_2(\bx)|^2\right)\,d\bx<\infty\right.\right\},
\end{equation*}
then the following embedding results hold.
\begin{lemma}\label{lem:compact} Under the assumption that $V_j(\bx)\ge0$ ($j=1,2$)
for $\bx\in {\mathbb R}^d$ are confining potentials,
i.e. $\lim\limits_{|\bx|\to\infty}V_j(\bx)=\infty$ ($j=1,2$),
we have that
the embedding $X\hookrightarrow L^{p_1}(\mathbb R^d)\times L^{p_2}(\mathbb R^d)$
is compact provided that exponents $p_1$
and $p_2$ satisfy
\begin{equation}
\begin{cases}
p_1,p_2\in[2,6),\quad d=3,\\
p_1,p_2\in[2,\infty),\quad d=2,\\
p_1,p_2\in [2,\infty],\quad d=1.
\end{cases}
\end{equation}
\end{lemma}

Then for the existence and uniqueness of the problem
(\ref{eq:minimize}) or (\ref{eq:minimize1}), we have

\begin{theorem}\label{thm:mres}(Existence and uniqueness)
Suppose  $V_j(\bx)\ge 0$ ($j=1,2$) satisfying
$\lim\limits_{|\bx|\to\infty}V_j(\bx)=\infty$, then  there exists a
minimizer
 $\Phi_g=(\phi_1^g,\phi_2^g)^T\in S$ of
 (\ref{eq:minimize}) if one of the following conditions holds,
 \begin{enumerate}\renewcommand{\labelenumi}{(\roman{enumi})}
\item $d=3$ and  the matrix $A$  is
either semi-positive definite or nonnegative.
\item $d=2$, $\beta_{11}>-C_b$, $\beta_{22}>-C_b$ and $\beta_{12}\ge-C_b-\sqrt{(C_b+\beta_{11})(C_b+\beta_{22})}$.
\item $d=1$.
\end{enumerate}
In addition, $e^{i\theta_0}\Phi_g$ is also a ground state
of \eqref{eq:minimize} for any $\theta_0\in\mathbb [0,2\pi)$.
In particular, when $\Omega=0$ the ground state is unique  up to a constant phase factor if the
matrix $A$ is semi-positive definite
and $I(\bx)\not\equiv0$ \eqref{eq:Ibx}.
In contrast, there exists no ground state of \eqref{eq:minimize} if one of the following holds
 \begin{enumerate}\renewcommand{\labelenumi}{(\roman{enumi})}
\item $d=3$ $\beta_{11}<0$ or $\beta_{22}<0$ or $\beta_{12}<0$ with $\beta_{12}^2>\beta_{11}\beta_{22}$;

\item $d=2$, $\beta_{11}<-C_b$ or $\beta_{22}<-C_b$ or $\beta_{12}<-C_b-\sqrt{(C_b+\beta_{11})(C_b+\beta_{22})}$.
\end{enumerate}

\end{theorem}

\smallskip
\begin{proof} The proof is similar to that for the case when $k_0=0$ in \cite{Bao2009} via using
the formulation (\ref{eq:minimize1}) and the details are omitted here for brevity.
\end{proof}

\subsection{Properties in different limiting parameter regimes}
From now on, we assume the conditions  for the existence of ground states in Theorem \ref{thm:mres} hold.
Introducing an auxiliary  energy functional $\tilde{E}_0(\tilde{\Phi})$ for $\tilde{\Phi}=(\tilde{\phi}_1,\tilde{\phi}_2)^T$
\begin{eqnarray}\label{eq:mini1} \tilde{E}_0(\tilde{\Phi})&=&\int_{{\mathbb R}^d}\biggl[
\sum\limits_{j=1}^2\left(\frac12|\nabla\tilde{\phi}_j|^2+V_j(\bx)|\tilde{\phi}_j|^2\right)+\frac{\delta}{2}
(|\tilde{\phi}_1|^2-|\tilde{\phi}_2|^2)+\frac{\beta_{11}}{2}|\tilde{\phi}_1|^4+\frac{\beta_{22}}{2}|\tilde{\phi}_2|^4
\nonumber\\
&&\qquad
+\beta_{12}|\tilde{\phi}_1|^2|\tilde{\phi}_2|^2\biggl]d\bx=\tilde{E}(\tilde{\Phi})-\Omega \int_{{\mathbb R}^d}\text{Re}
(e^{i2k_0x}\tilde{\phi}_1\overline{\tilde{\phi}}_2)d\bx,
\end{eqnarray}
 we know that the nonconvex minimization problem
\be \label{min987}
\tilde E_g^{(0)}:=\tilde{E}_0(\tilde\Phi_g^{(0)})=\min_{\tilde\Phi\in S}\tilde{E}_0(\tilde\Phi),
\ee
admits a unique positive minimizer $\tilde\Phi_g^{(0)}=(\tilde{\phi}_1^{g,0},\tilde{\phi}_2^{g,0})^T\in S$  if the matrix $A$
is semi-positive definite and $I(\bx)\not\equiv0$ \eqref{eq:Ibx} \cite{Bao2009}. For a given $k_0\in {\mathbb R}$, let
$\tilde{\Phi}^{k_0}=(\tilde{\phi}_1^{k_0},\tilde{\phi}_2^{k_0})^T\in S$ be a ground state of
\eqref{eq:minimize1} when all other parameters are fixed, then we have

\begin{theorem}\label{thm:k0change1}(Large $k_0$ limit).  Suppose the matrix $A$ is semi-positive definite
and $I(\bx)\not\equiv0$ \eqref{eq:Ibx}. When $k_0\to\infty$, we have that the ground state
$\tilde{\Phi}^{k_0}=(\tilde{\phi}_1^{k_0},\tilde{\phi}_2^{k_0})^T$ of
\eqref{eq:minimize1} converges to
a ground state of (\ref{min987})
in  $L^{p_1}\times L^{p_2}$ sense with
$p_1,p_2$  given in Lemma \ref{lem:compact}, i.e., there exist constants $\theta_{k_0}\in [0,2\pi)$
such that $e^{i\theta_{k_0}}(\tilde{\phi}_1^{k_0},\tilde{\phi}_2^{k_0})^T$ converge to the unique positive ground state
$\tilde\Phi_g^{(0)}$ of (\ref{min987}).  In other words,
large $k_0$ in the CGPEs (\ref{eq:cgpe199:sec9}) will remove the effect of Raman coupling $\Omega$, i.e. large $k_0$ limit
is effectively letting $\Omega\to0$.
\end{theorem}

\begin{proof} Let $\tilde{\Phi}^{k_0}=(\tilde{\phi}_1^{k_0},\tilde{\phi}_2^{k_0})^T\in S$ be a ground state of \eqref{eq:minimize1}, then we have
\be
\tilde{E}(\tilde{\Phi}^{k_0})\leq \tilde{E}^{(0)}_g=\min\limits_{\tilde{\Phi}^\in S}\tilde{E}_0(\tilde{\Phi}),
\ee
where $\tilde{E}^{(0)}_g$ is attained at the unique positive ground state of $\tilde{E}_0(\cdot)$ in
\eqref{min987}.

Under the condition of the theorem, we know that $(\tilde{\phi}_1^{k_0},\tilde{\phi}_2^{k_0})^T\in S$
 is a bounded sequence in $X$. Hence, for any sequence $\{k_0^m\}_{m=1}^\infty$ with $k_0^m\to\infty$, there exists
a subsequence $(\tilde{\phi}_1^{k_0^m},\tilde{\phi}_2^{k_0^m})^T$ (denote as the original sequence for simplicity) such that
\be
(\tilde{\phi}_1^{k_0^m},\tilde{\phi}_2^{k_0^m})^T\hookrightarrow (\tilde{\phi}_1^\infty,\tilde{\phi}_2^\infty)^T\in X, \text{weakly}.
\ee
Lemma \ref{lem:compact} ensures that such convergence is strong in $L^{p_1}\times L^{p_2}$. In particular, we get
\be
\tilde{E}_0(\tilde{\phi}_1^\infty,\tilde{\phi}_2^\infty)\leq\liminf\limits_{k_0^m\to\infty}\tilde{E}_0(\tilde{\phi}_1^{k_0^m},\tilde{\phi}_2^{k_0^m}).
\ee
and $(\tilde{\phi}_1^\infty,\tilde{\phi}_2^\infty)^T\in S$. Recalling that
\begin{equation*}
\begin{split}
&\Omega\int_{\mathbb R^d}\text{Re}(e^{2k_0^mix}\tilde{\phi}_1^{k_0^m}\overline{\tilde{\phi}^{k_0^m}_2})\,d\bx\\
&=
\Omega\int_{\mathbb R^d}\text{Re}(e^{2k_0^mix}(\tilde{\phi}_1^{k_0^m}-\tilde{\phi}_1^\infty)\overline{\tilde{\phi}^{k_0^m}_2})\,d\bx
+\Omega\int_{\mathbb R^d}\text{Re}(e^{2k_0^mix}\tilde{\phi}_1^\infty(\overline{\tilde{\phi}^{k_0^m}_2}-
\overline{\tilde{\phi}^{\infty}_2}))\,d\bx\\
&\quad+\Omega\int_{\mathbb R^d}\text{Re}(e^{2k_0^mix}\tilde{\phi}_1^\infty\overline{\tilde{\phi}_2^\infty})\,d\bx,
\end{split}
\end{equation*}
using the $L^{p_1}\times L^{p_2}$ convergence of $(\tilde{\phi}_1^{k_0^m},\tilde{\phi}_2^{k_0^m})^T$  and Riemann-Lebesgue Lemma, we deduce
\be
\lim_{k_0^m\to\infty}\Omega\int_{\mathbb R^d}\text{Re}(e^{2k_0^mix}\tilde{\phi}_1^{k_0^m}\overline{\tilde{\phi}_2^{k_0^m}})\,d\bx=0.
\ee
Hence,
\be
\tilde{E}_0(\tilde{\phi}_1^\infty,\tilde{\phi}_2^\infty)\leq \liminf\limits_{k_0^m\to\infty}\tilde{E}_0(\phi_1^{k_0^m},\phi_2^{k_0^m})\leq
\liminf\limits_{k_0^m\to\infty}\tilde{E}(\phi_1^{k_0^m},\phi_2^{k_0^m})\leq E_g^{(0)}.
\ee
This means $(\phi_1^\infty,\phi_2^\infty)^T\in S$ is also a minimizer of the
energy \eqref{eq:mini1} in the nonconvex set $S$. The rest then follows
from the fact that the ground state of \eqref{eq:mini1} is unique up to a constant phase factor.
\end{proof}
\begin{remark} Under the assumption of Theorem \ref{thm:k0change1} and
 $\Omega=o(|k_0|)$ as $k_0\to \infty$,  the conclusion
of Theorem \ref{thm:k0change1} still holds (see details in Theorem \ref{thm:order}).
In fact, Theorem \ref{thm:k0change1} holds when the matrix $A$ is nonegative, but
the limiting profile is non-unique since there is no uniqueness for the positive ground state $\Phi_g^{(0)}$ of \eqref{min987} \cite{Bao2009}.
\end{remark}

Then, we  conclude the following for the ground state of CGPEs \eqref{eq:cgpe1} given by the minimization problem \eqref{eq:minimize} when $k_0\to\infty$.
\begin{theorem}\label{thm:k0change}(Large $k_0$ limit).  Suppose the matrix $A$ is semi-positive definite and $I(\bx)\not\equiv0$ \eqref{eq:Ibx}. When  $k_0\to\infty$,
the ground state   $\Phi^{k_0}_g=(\phi_1^{g},\phi_2^g)^T$  of
\eqref{eq:minimize} corresponds to a ground state $\tilde{\Phi}^{k_0}_g=(e^{ik_0x}\tilde{\phi}_1^{g,0},e^{-ik_0x}\tilde{\phi}_2^{g,0})^T$
of \eqref{eq:minimize1} (see \eqref{eq:equiv}),
where
$\tilde{\Phi}_g^{k_0}$ converges to a ground state of \eqref{min987}, i.e. for some $\theta_{k_0}\in \mathbb R$,
$e^{i\theta_{k_0}}(e^{-ik_0x}\phi_1^g,e^{ik_0x}\phi_2^g)^T$ converge to
the positive ground state  $\tilde{\Phi}_g^{(0)}$  of \eqref{min987}   in $L^{p_1}\times L^{p_2}$ sense, where $p_1,p_2$ are given in Lemma \ref{lem:compact}.
 In other words,
large $k_0$
will remove the effect of Raman coupling $\Omega$ in the CGPEs \eqref{eq:cgpe1}.
\end{theorem}

Analogous to the case of the two-component BEC without SO coupling \cite{Bao2009}, i.e. $k_0=0$, we have the following results.
\begin{theorem}\label{thm:ogchange}(Large $\Omega$ limit). Suppose the matrix $A$ is either semi-positive definite or nonnegative.
When $|\Omega|\to\infty$, the ground state $\Phi_g$ of \eqref{eq:minimize} converges to a state
$(\phi_g,\text{sgn}(-\Omega)\phi_g)^T$ in $L^{p_1}\times L^{p_2}$ sense, where $p_1,p_2$ are given in Lemma \ref{lem:compact}, i.e., large $\Omega$ will
 remove the effect of $k_0$ in the CGPEs \eqref{eq:cgpe1}. Here $\phi_g$ minimizes the following energy under the constraint
$\|\phi_g\|:=\int_{{\Bbb R}^d} |\phi_g(\bx)|^2d\bx=1/\sqrt{2}$,
\begin{align}\label{eq:oglim}
E_s(\phi)&=\int_{{\mathbb R}^d}\biggl[
\frac12|\nabla\phi|^2+\frac{V_1(\bx)+V_2(\bx)}{2}|\phi|^2+
\frac{\beta_{11}+\beta_{22}+2\beta_{12}}{4}|\phi|^4\biggl]d\bx,
\end{align}
where $\phi_g$ is unique up to a constant phase shift and can be chosen as strictly positive.
\end{theorem}

\begin{theorem} (Large $\delta$ limit). Assume the matrix $A$ is either semi-positive definite or nonnegative.
When  $\delta \to +\infty$,
the ground state $\Phi^g$ of \eqref{eq:minimize} converges to a state
$(0,\phi_g)^T$ in $L^{p_1}\times L^{p_2}$ sense, where $p_1,p_2$ are given in Lemma \ref{lem:compact}.  Here $\phi_g$ minimizes the following energy under the constraint
$\|\phi_g\|=1$,
\begin{align*}
E_1(\phi)=\int_{{\mathbb R}^d}\biggl[
\frac12|\nabla\phi|^2+V_2(\bx)|\phi|^2-ik_0\bar{\phi}\p_x\phi+\frac{\beta_{22}}{2}|\phi|^4\biggl]d\bx,
\end{align*}
and such $\phi_g$ is unique up to a constant phase shift.
When  $\delta \to -\infty$, the ground state $\Phi_g$ of \eqref{eq:minimize} converges to a state
$(\varphi_g,0)^T$, where $\varphi_g$ minimize the following energy under the constraint
$\|\varphi_g\|_2=1$,
\begin{align*}
E_2(\varphi)=\int_{{\mathbb R}^d}\biggl[
\frac12|\nabla\varphi|^2+V_1(\bx)|\varphi|^2+ik_0\bar{\varphi}\p_x\varphi+\frac{\beta_{11}}{2}|\varphi|^4\biggl]d\bx,
\end{align*}
and such $\varphi_g$ is unique up to a constant phase shift.
\end{theorem}

\subsection{Convergence rate}
From the discussion in the previous section, we find that the appearance of SO coupling term $k_0$ causes a new transition in the ground states
of the CGPEs \eqref{eq:cgpe1} \cite{Bao2009}.
When $k_0=0$, i.e. there is no SO coupling,  the ground state $\Phi_g=(\phi_1^g,\phi_2^g)^T$ of \eqref{eq:minimize} can be chosen as real functions
 $\phi_1^g=|\phi_1^g|$ and $\phi_2^g=-\text{sgn}(\Omega)|\phi_2^g|$ \cite{Bao2009}. When $k_0\to\infty$,
 $\tilde{\Phi}_g=(e^{-ik_0x}\phi_1^g,e^{ik_0x}\phi_2^g)$ of \eqref{eq:equiv} will
 converge  to
 the ground state of \eqref{min987} (see Theorem \ref{thm:k0change}), i.e. it is equivalent to let $\Omega=0$ in the large $k_0$ limit.
 Here, we are going to characterize the convergence rates of the ground state  $\Phi_g$ of \eqref{eq:minimize} in the above two cases,
 i.e. $k_0\to0$ and $k_0\to\infty$.

For small $k_0$,  it is convenient to rewrite the energy \eqref{eq:energy} for $\Phi=(\phi_1,\phi_2)^T$ as
\begin{eqnarray} E(\Phi)&=&\int_{{\mathbb R}^d}\biggl[
\sum\limits_{j=1}^2\left(\frac12|(\nabla+i(3-2j)k_0{\bf e}_x)\phi_{j}|^2+V_j(\bx)|\phi_j|^2\right)+\frac{\delta}{2}(|\phi_1|^2-|\phi_2|^2)
\nonumber\\&&\label{eq:Erf}
+\frac{\beta_{11}}{2}|\phi_1|^4+\frac{\beta_{22}}{2}|\phi_2|^4
+\beta_{12}|\phi_1|^2|\phi_2|^2+\Omega\cdot\text{Re}
(\phi_1\bar{\phi}_2)\biggl]d\bx-k_0^2\|\Phi\|^2,
\end{eqnarray}
where ${\bf e}_x$ is the unite vector of $x$ axis,
and  we denote
\begin{eqnarray*} E_0(\Phi)=E(\Phi)-\int_{\mathbb R^d}\left(ik_0\overline{\phi}_1\p_x\phi_1-ik_0\overline{\phi}_2\p_x\phi_2\right)\,d\bx,
\end{eqnarray*}
with $E_0(\cdot)$ being the energy of the CGPEs \eqref{eq:cgpe1} when $k_0=0$.

Without loss of generality, we   assume $\Omega<0$.
\begin{theorem}\label{thm:error}Suppose $\Omega<0$, $\lim\limits_{|\bx|\to\infty}V_j(\bx)=\infty$ ($j=1,2$)
and the matrix $A$ is semi-positive definite.
Denoting $\widehat{\Phi}_g=(\varphi_1^g,\varphi_2^g)^T\in S$ as the unique nonnegative ground state of $E_0(\Phi)$ in $S$ \cite{Bao2009},
there exists a constant $C>0$ independent of $k_0$ such that the  ground state $\Phi_g=(\phi_1^g,\phi_2^g)^T\in S$  of \eqref{eq:minimize}
satisfies
\be
\||\phi_1^g|-\varphi_1^g\|+\||\phi_2^g|-\varphi_2^g\|\leq C|k_0|.
\ee
\end{theorem}
\begin{proof}
First of all, recalling \eqref{eq:mini1} and \eqref{eq:Erf}, we have the lower bound of $E_g=E(\phi_1^g,\phi_2^g)$ as \cite{Bao2009,LiebL}
\be
E(\phi_1^g,\phi_2^g)\ge \tilde{E}_0(|\phi_1^g|,|\phi_2^g|)-|\Omega|\int_{\mathbb R^d}|\phi_1^g||\phi_2^g|d\bx
-\frac{k_0^2}{2}= E_0(|\phi_1^g|,|\phi_2^g|)-\frac{k_0^2}{2},
\ee
and the upper bound
\be
E(\phi_1^g,\phi_2^g)\leq E(\varphi_1^g,\varphi_2^g)=E_0(\varphi_1^g,\varphi_2^g).
\ee
Hence,
\be
E_0(|\phi_1^g|,|\phi_2^g|)-E_0(\varphi_1^g,\varphi_2^g)\leq \frac{k_0^2}{2}.
\ee
 In addition, $(\varphi_1^g,\varphi_2^g)^T\in S$ satisfies the nonlinear eigenvalue problem
\be
\begin{split} &\mu_1 \varphi_{1}^g=\left[-\frac{1}{2}\nabla^2
+V_1(\bx)+\frac{\delta}{2}
+(\beta_{11}|\varphi_{1}^g|^2+\beta_{12}|\varphi_{2}^g|^2)\right]\varphi_{1}^g+\frac{\Omega}{2}
\varphi_{2}^g, \\
&\mu_1 \varphi_{2}^g=\left[-\frac{1}{2}\nabla^2
+V_2(\bx)-\frac{\delta}{2}+(\beta_{12}|\varphi_{1}^g|^2+\beta_{22}|\varphi_{2}^g|^2)\right]\varphi_{2}^g
+\frac{\Omega}{2}
\varphi_{1}^g,\end{split} \ee
where $\mu_1$ is the corresponding eigenvalue (or chemical potential). For this nonlinear eigenvalue problem,
we denote the linearized operator $L$  acting on $\Phi=(\phi_1,\phi_2)^T $ as
\be
L\Phi=\begin{pmatrix}
L_1&\frac{\Omega}{2}\\
\frac{\Omega}{2}&L_2\end{pmatrix}\Phi,\quad
L_j=-\frac{1}{2}\nabla^2
+V_j(\bx)+\frac{\delta}{2}(3-2j)
+\sum\limits_{l=1}^2\beta_{jl}|\varphi_{l}^g|^2,\quad j=1,2.
\ee
It is clear that $(\varphi_1^g,\varphi_2^g)^T$ is an eigenfunction of $L$ with eigenvalue $\mu_1$ and by the nonnegativity of
$(\varphi_1^g,\varphi_2^g)^T$, $\mu_1$ is the smallest eigenvalue. In fact, the eigenfunctions $(\varphi_1^k,\varphi_2^k)^T\in S$
($k=1,2,\ldots$)
of $L$
corresponds to eigenvalue $\mu_k$ which can be arranged  in the nondecreasing order, i.e. $\mu_k$ is nondecreasing. The eigenfunctions form
an orthonormal basis
of $L^2(\mathbb R^d)\times L^2(\mathbb R^d)$ and $\mu_1<\mu_2$ with $(\varphi_1^g,\varphi_2^g)^T=(\varphi_1^1,\varphi_2^1)$ (positive ground state is unique).

Denoting $\Phi_{e}=(\phi_1^e,\phi_2^e)^T:=(|\phi_1^g|-\varphi_1^g,|\phi_2^g|-\varphi_2^g)$,
and using the Euler-Lagrange equation for $(\varphi_1^g,\varphi_2^g)^T\in S$, we find
\begin{align*}
E_0(|\phi_1^g|,|\phi_2^g|)=&\int_{\mathbb R^d}\bigg(
\sum\limits_{j=1}^2\frac{\beta_{jj}}{2}(|\phi_j^g|^2-|\varphi_j^g|^2)^2
+\beta_{12}(|\phi_1^g|^2-|\varphi_1^g|^2)(|\phi_2^g|^2-|\varphi_2^g|^2)\bigg)\,d\bx\\
&+E_0(\varphi_1^g,\varphi_2^g)+\int_{\mathbb R^d}\Phi_{e}^TL\Phi_e\,d\bx-\mu_1\|\Phi_e\|^2.
\end{align*}
Using the fact that $L+c$ ($c\ge0$ sufficiently large) induces an equivalent norm
in $X$, we can take expansion $(\phi_1^e,\phi_2^e)^T=\sum\limits_{k=1}^\infty c_k(\varphi_1^k,\varphi_2^k)^T$ with $
\sum\limits_{k=1}^\infty c_k^2=\|\Phi_e\|^2$, and estimate
\begin{align*}
\int_{\mathbb R^d}\Phi_e^TL\Phi_e\,d\bx=\sum\limits_{k=1}^\infty\mu_k c_k^2\ge\mu_1c_1^2+\mu_2(\|\Phi_e\|^2-c_1^2),
\end{align*}
with $c_1=\frac12\|\Phi^e\|^2=\frac{1}{2}(\||\phi_1^g|-\varphi_1^g\|^2+\||\phi_2^g|-\varphi_2^g\|^2)<1$. Hence, we obtain
\begin{align*}
E_0(|\phi_1^g|,|\phi_2^g|)- E_0(\varphi_1^g,\varphi_2^g)
\ge
(\mu_2-\mu_1)(2c_1-c_1^2)\ge(\mu_2-\mu_1)c_1.
\end{align*}
Since the gap $\mu_2-\mu_1$ is independent of $k_0$, we draw the conclusion.
\end{proof}

For large $k_0$, we have the similar results.
\begin{theorem} Suppose $\Omega<0$, $\lim\limits_{|\bx|\to\infty}V_j(\bx)=\infty$ ($j=1,2$)
 the matrix $A$ is  semi-positive definite and $I(\bx)\not\equiv0$.
Denoting $\tilde{\Phi}_g^{(0)}=(\tilde{\phi}_1^{g,0},\tilde{\phi}_2^{g,0})^T\in S$ as
the unique nonnegative ground state of \eqref{min987} (minimizer of $\tilde{E}_0(\cdot)$ of \eqref{eq:mini1} in $S$),
there exists a constant $C>0$ independent of $k_0$ such that the ground state $\Phi_g=(\phi_1^g,\phi_2^g)^T\in S$ of \eqref{eq:minimize}
satisfies
\be
\||\phi_1^g|-\tilde{\phi}_1^{g,0}\|+\||\phi_2^g|-\tilde{\phi}_2^{g,0}\|\leq C/\sqrt{k_0}.
\ee
\end{theorem}
\begin{proof} From \eqref{eq:equiv}, we know $\tilde{\Phi}_g=(\tilde{\phi}_1^g,\tilde{\phi}^g_2)^T=(e^{-ik_0x}\phi_1^g,e^{ik_0x}\phi_2^g)^T$ minimizes the energy $\tilde{E}$ in \eqref{eq:minimize1}.
Noticing
\begin{equation*}
\begin{split}
\Omega\int_{\mathbb R^d}\text{Re}(e^{2k_0ix}\tilde{\phi}_1^{g}\overline{\tilde{\phi}^{g}_2})\,d\bx&
=\frac{-\Omega}{2k_0}\int_{\mathbb R^d}\text{Re}\left(ie^{2k_0ix}(\p_x\tilde{\phi}_1^{g}\overline{\tilde{\phi}^{g}_2}+ie^{2k_0ix}(\tilde{\phi}_1^{g}\p_x\overline{\tilde{\phi}^{g}_2}\right)\,d\bx\\
&\ge-\vep(\|\p_x\tilde{\phi}_1^{g}\|^2+\|\p_x\tilde{\phi}_2^g\|^2)+\frac{\Omega^2}{4\vep k_0^2}(\|\tilde{\phi}_1^{g}\|^2+\|\tilde{\phi}_2^{g}\|^2),\quad \vep >0,
\end{split}
\end{equation*}
 we find
\begin{equation*}
\tilde{E}(\tilde{\phi}_1^g,\tilde{\phi}_2^g)\ge \tilde{E}_0(\tilde{\phi}_1^g,\tilde{\phi}_2^g)-\frac14(\|\p_x\tilde{\phi}_1^{g}\|^2
+\|\p_x\tilde{\phi}_2^g\|^2)-\frac{\Omega^2}{k_0^2}.
\end{equation*}
On the other hand, we have
\begin{equation*}
\tilde{E}(\tilde{\phi}_1^g,\tilde{\phi}_2^g)\leq \tilde{E}(\tilde{\phi}_1^{g,0},\tilde{\phi}_2^{g,0})\leq
\tilde{E}_0(\tilde{\phi}_1^{g,0},\tilde{\phi}_2^{g,0})+\frac{C_1|\Omega|}{k_0},
\end{equation*}
where $C_1>0$ is a constant. Thus, we know
\begin{equation*}
\|\tilde{\Phi}^g\|_X^2\leq C(1+\Omega^2/k_0^2),
\end{equation*}
and it follows that  for large $k_0$,
\begin{equation*}
\tilde{E}(\tilde{\phi}_1^g,\tilde{\phi}_2^g)\ge \tilde{E}_0(\tilde{\phi}_1^g,\tilde{\phi}_2^g)-C_2\frac{|\Omega|}{|k_0|}\|\tilde{\Phi}^g\|_X\ge
\tilde{E}_0(\tilde{\phi}_1^g,\tilde{\phi}_2^g)-C_3\frac{|\Omega|}{|k_0|},
\end{equation*}
where $C_2$ and $C_3$ are two positive constant. We then conclude
\begin{equation*}
\tilde{E}_0(|\tilde{\phi}_1^g|,|\tilde{\phi}_2^g|)\leq \tilde{E}_0(\tilde{\phi}_1^g,\tilde{\phi}_2^g)
\leq \tilde{E}(\tilde{\phi}_1^g,\tilde{\phi}_2^g)+\frac{C_3|\Omega|}{|k_0|}\leq \tilde{E}_0(\tilde{\phi}_1^{g,0},\tilde{\phi}_2^{g,0})+\frac{C|\Omega|}{k_0}.
\end{equation*}
The rest of the proof is similar to that in Theorem \ref{thm:error} and is omitted  here.
\end{proof}

\subsection{Competition between $\Omega$ and $k_0$}
In the previous subsection, we find that large Raman coupling $\Omega$ will remove the effect of SO coupling  $k_0$
in  the asymptotic profile of the ground states of \eqref{eq:minimize}
and the reverse is true, i.e.
there is a competition between these two parameters. Here, we are going to study how the relation
between $k_0$ and $\Omega$ affects the ground state profile of \eqref{eq:minimize}. The results
are summarized as follows.
\begin{theorem}\label{thm:order} Suppose $\lim\limits_{|\bx|\to\infty}V_j(\bx)=\infty$ ($j=1,2$),
 the matrix $A$ is either semi-positive definite or nonnegative, then
we have

(i) If $|\Omega|/|k_0|^2\gg1$, $|\Omega|\to\infty$, the ground state $\Phi_g=(\phi^g_1,\phi^g_2)^T$
of \eqref{eq:minimize} for the CGPEs \eqref{eq:cgpe1} converges
to a state $(\phi_g,\text{sgn}(-\Omega)\phi_g)^T$, where $\phi_g$
minimizes the energy \eqref{eq:oglim} under the constraint $\|\phi_g\|=1/\sqrt{2}$, i.e.
conclusion of Theorem \ref{thm:ogchange} holds.

(ii) If $|\Omega|/|k_0|\ll1$, $|k_0|\to\infty$, the ground state $\Phi_g=(\phi^g_1,\phi^g_2)^T$ of
\eqref{eq:minimize} for the CGPEs \eqref{eq:cgpe1} converges
to a state $(e^{-ik_0x}\tilde{\phi}_1^{g,0},e^{ik_0x}\tilde{\phi}_2^{g,0})^T$,
where $\widetilde{\Phi}_g^{(0)}=(\tilde{\phi}_1^{g,0},\tilde{\phi}_2^{g,0})^T$
is a ground state of  \eqref{min987} for the energy $E_s(\cdot)$ in \eqref{eq:mini1},
i.e.,  conclusion of Theorem \ref{thm:k0change} holds.

(iii) If $|k_0|\ll|\Omega|\ll |k_0|^2$ and $|k_0|\to\infty$, the
leading order of the ground state energy $E_g:=E(\Phi_g)$ of \eqref{eq:minimize} for the CGPEs
\eqref{eq:cgpe1} is given by
 $E_g=-\frac{k_0^2}{2}-C_0\frac{|\Omega|^2}{|k_0|^2}+o\left(\frac{|\Omega|^2}{|k_0|^2}\right)$,
 where $C_0>0$ is a generic constant.

\end{theorem}
\begin{proof} Without loss of generality, we assume $\Omega<0$.

 (i) It is obvious that $\Phi_g$ also  minimizes
the following energy for $\Phi=(\phi_1,\phi_2)^T\in S$
\begin{eqnarray*} E(\Phi)&=&-\frac{|\Omega|}{2}+\int_{{\mathbb R}^d}\biggl[
\sum\limits_{j=1}^2\left(\frac12|\nabla\phi_{j}|^2+V_j(\bx)|\phi_j|^2\right)+\frac{\delta}{2}(
|\phi_1|^2-|\phi_2|^2)+ik_0\bar{\phi}_1\p_x\phi_1\\
&&-ik_0\bar{\phi}_2\p_x\phi_2+\frac{\beta_{11}}{2}|\phi_1|^4+\frac{\beta_{22}}{2}|\phi_2|^4
+\beta_{12}|\phi_1|^2|\phi_2|^2+\frac{|\Omega|}{2}
|\phi_1-\phi_2|^2\biggl]d\bx.
\end{eqnarray*}
A simple choice of testing state $(\phi_g,\phi_g)^T\in S$ shows that $E(\cdot)+\frac{|\Omega|}{2}$ is uniformly bounded from above, i.e.
\be
E_g+\frac{|\Omega|}{2}=E(\Phi_g)+\frac{|\Omega|}{2}\leq E(\phi_g,\phi_g)+\frac{|\Omega|}{2}=2E_s(\phi_g):=2E_{s}^g.
\ee
To get a lower bound for $E_g$, using Cauchy inequality, we  have for any $\vep>0$,
\begin{align*}
\int_{\mathbb R^d}ik_0\left(\bar{\phi}_1\p_x\phi_1-\bar{\phi}_2\p_x\phi_2\right)\,d\bx
=&\int_{\mathbb R^d}ik_0\left[(\bar{\phi}_1-\bar{\phi}_2)\p_x\phi_1-(\phi_1-\phi_2)\p_x\bar{\phi}_2\right]\,d\bx\\
\ge&-\frac{\vep}{2}(\|\p_x\phi_1\|^2+\|\p_x\phi_2\|^2)-\frac{k_0^2}{2\vep}\|\phi_1-\phi_2\|^2.
\end{align*}
Hence, by setting $\vep=1$ in the above inequality and recalling $\|\phi_1-\phi_2\|\leq \sqrt{2}$ for $\Phi=(\phi_1,\phi_2)^T\in S$,
we bound $E_g$ from below by
\be
E_g+\frac{|\Omega|}{2}\ge -\frac{|\delta|}{2}-\frac{k_0^2}{2}+\frac{|\Omega|}{2}\|\phi_1^g-\phi_2^g\|^2.
\ee
Combining the upper and lower bounds of $E_g+\frac{|\Omega|}{2}$, we get
\be
\|\phi_1^g-\phi_2^g\|\leq \frac{4E_{s}^g+|\delta|}{|\Omega|}+\frac{k_0^2}{|\Omega|}.
\ee
If $k_0^2/|\Omega|=o(1)$ and $|\Omega|\to\infty$, we see $\phi_1^g-\phi_2^g\to 0$ in $L^2$ and the ground state sequence $\Phi^g=(\phi_1^g,\phi_2^g)^T$
is bounded in $X$. Analogous to the proof  in Theorem \ref{thm:k0change} and \cite{Bao2009}, we can draw the conclusion and the detail
is omitted here.

(ii) It is equivalent to prove that in this case, the ground state $\tilde{\Phi}_g=(\tilde{\phi}^g_1,\tilde{\phi}^g_2)^T=(e^{-ik_0x}\phi^g_1,e^{ik_0x}\phi^g_2)^T
 $of \eqref{eq:minimize1} converges to the ground state
of \eqref{min987}.
Using integration by parts and Cauchy inequality, we get
\begin{equation}\label{eq:ch1}
\begin{split}
\Omega\int_{\mathbb R^d}\text{Re}(e^{i2k_0x}\tilde{\phi}_1^{g}\overline{\tilde{\phi}^{g}_2})\,d\bx&=\frac{\Omega}{2k_0}\int_{\mathbb R^d}\text{Re}\left(ie^{i2k_0x}
\left(\p_x\tilde{\phi}_1^{g}\overline{\tilde{\phi}^{g}_2}+\tilde{\phi}_1^{g}
\p_x\overline{\tilde{\phi}^{g}_2}\right)\right)\,d\bx\\
&\ge-\frac{|\Omega|}{2|k_0|}(\|\p_x\tilde{\phi}_1^{g}\|\,\|\tilde{\phi}_2^{g}\|+
\|\p_x\tilde{\phi}_2^{g}\|\,\|\tilde{\phi}_1^{g}\|).
\end{split}
\end{equation}
Having this in hand, we could proceed as in the proof of Theorem \ref{thm:k0change}.

 (iii)  Similar to the case of (ii),  we need only consider the ground state $\tilde{\Phi}_g=(\tilde{\phi}_1^g,\tilde{\phi}_2^g)^T\in S$ of \eqref{eq:minimize1}.
 Applying Cauchy inequality in \eqref{eq:ch1}, we have
\begin{equation}\label{eq:ch}
\begin{split}
\Omega\int_{\mathbb R^d}\text{Re}(e^{2k_0ix}\tilde{\phi}_1^{g}\overline{\tilde{\phi}^{g}_2})\,d\bx
\ge-\frac14\|\p_x\tilde{\phi}_1^{g}\|-\frac14\|\p_x\tilde{\phi}_2^{g}\|-\frac{2|\Omega|^2}{|k_0|^2}.
\end{split}
\end{equation}
By choosing sufficiently smooth (e.g. $H^3\cap X$ ) test states for $\tilde{E}(\cdot)$ and using
integration by parts as \eqref{eq:ch1}, it is straightforward to get the
upper bound
\be\label{eq:upp}
\tilde{E}(\tilde{\phi}_1^g,\tilde{\phi}_2^g)\leq C+\frac{|\Omega|}{|k_0|^3}.
\ee
Combining \eqref{eq:ch} and \eqref{eq:upp}, we find that
\be\label{eq:telower}
\tilde{E}(\tilde{\phi}_1^g,\tilde{\phi}_2^g)\ge C-\frac{2|\Omega|^2}{|k_0|^2},
\qquad\|\tilde{\Phi}_g\|_{X}^2\leq C+\frac{2|\Omega|^2}{|k_0|^2}+\frac{|\Omega|}{|k_0|^3}\leq C\frac{|\Omega|^2}{|k_0^2|}.
\ee
Then, it follows from \eqref{eq:ch1}  that
\be
\left|\Omega\int_{\mathbb R^d}\text{Re}(e^{2k_0ix}\tilde{\phi}_1^{g}\overline{\tilde{\phi}^{g}_2})\,d\bx\right|\leq \frac{|\Omega|}{|k_0|}\|\tilde{\Phi}_g\|_X=O\left(\frac{|\Omega|^2}{|k_0|^2}\right).
\ee
On the other hand, we can choose test states as follows. In one dimension, let $\rho(x)$ be a $C_0^\infty$ even real-valued function with $\|\rho\|=\sqrt{2}/2$ and we choose
\be
\tilde{\phi}_1(x)=N_\vep\rho(x)[1-\vep\cos(2k_0x)],\quad \tilde{\phi}_2(x)=\rho(x),
\ee
here $N_\vep$ is a normalization constant to ensure that $(\tilde{\phi}_1,\tilde{\phi}_2)^T\in S$ and it is clear that $N_\vep$ is close to 1 for small $\vep$ and large $k_0$.
Recalling $\tilde{E}_0(\cdot)$ in \eqref{eq:mini1}, we can calculate
\be\label{eq:te0}
\tilde{E}_0(\tilde{\phi}_1,\tilde{\phi}_2)=C_1+C_2\vep^2 |k_0|^2+o(\vep^2|k_0|^2),
\ee
and
\begin{equation*}
\begin{split}
\Omega\int_{\mathbb R}\text{Re}(e^{2k_0ix}\tilde{\phi}_1\overline{\tilde{\phi}}_2)\,dx&
=\Omega\int_{\mathbb R}\text{Re}\left(e^{2k_0ix}\rho^2(x)\right)dx
+\vep\Omega\int_{\mathbb R}\cos^2(2k_0x)\rho^2(x)dx\\\
&=\frac{\vep\Omega}{2}\int_{\mathbb R}\rho^2(x)\,dx+\frac{\Omega}{2}\int_{\mathbb R}\left[2\cos(2k_0x)+\cos(4k_0x)\right]\rho^2(x)dx,
\end{split}
\end{equation*}
where the second integral on the RHS is of arbitrary order at $O(|\Omega|/|k_0|^m)$ ($m\ge0$) by
using integration by parts and the property of $\rho(x)$.
Hence, we find
\be\label{eq:og1}
\Omega\int_{\mathbb R}\text{Re}(e^{2k_0ix}\tilde{\phi}_1\overline{\tilde{\phi}}_2)\,dx=-\frac{|\Omega|\vep}{2}+
o\left(\frac{|\Omega|}{|k_0|^3}\right).
\ee
Now, we get from \eqref{eq:mini1}, \eqref{eq:te0} and \eqref{eq:og1} that
\be
\tilde{E}(\tilde{\phi}_1,\tilde{\phi}_2)=C_1+C_2\vep^2 |k_0|^2-|\Omega|\vep/4+o(|\Omega|/|k_0|^3).
\ee
Since $|\Omega|\ll |k_0|^2$, we can choose $\vep=\gamma|\Omega|/|k_0|^2$ and $\gamma>0$ be sufficiently
small such that
 the term $C_2\vep^2|k_0|^2=C_2\gamma|\Omega|\vep\leq \frac{|\Omega|\vep}{8}$. So, we arrive at
 \be\label{eq:telupper}
 \tilde{E}(\tilde{\phi}_1^g,\tilde{\phi}_2^g)\leq \tilde{E}(\tilde{\phi}_1,\tilde{\phi}_2)\leq C-\frac{|\Omega|^2\gamma}{8|k_0|^2}+o\left(\frac{|\Omega|^2}{|k_0|^2}\right).
 \ee
 In two and three dimensions, similar constructions will show the same estimates.
 Thus the conclusion is an immediate consequence of \eqref{eq:reform},
 \eqref{eq:telower} and \eqref{eq:telupper}.
\end{proof}
\begin{remark}
For $|k_0|\ll|\Omega|\ll |k_0|^2$, the ground state of \eqref{eq:minimize} is
much more complicated. In such situation,
the above theorem shows that  oscillation of ground state densities
may occur at the order of  $O(|\Omega|/|k_0|^2)$ in amplitude
and $k_0$ in frequency. Such density oscillation is predicted in the physics literature
\cite{Li2012,Li2013}, known as the density modulation. It is of great interest
to identify the constant $C_0$ in  the conclusion (iii).
\end{remark}

\section{Numerical methods and results}\label{sec:nm}
In this section, we present efficient and accurate numerical methods for computing the
ground states based on (\ref{eq:minimize}) (or (\ref{eq:minimize1})) and dynamics based on the CGPEs
\eqref{eq:cgpe1}  (or (\ref{eq:cgpe199:sec9})) for the SO-coupled BEC.

\subsection{For computing ground states}
Let $t_n=n\tau$ ($n=0,1,2,\ldots$) be the time steps with $\tau>0$ as time step.
In order to compute the ground state
$\Phi_g=(\phi_1^g,\phi_2^g)^T$ of \eqref{eq:minimize} for a SO-coupled BEC,
we propose the following gradient flow with discrete normalization (GFDN), which is widely used in computing
the ground states of BEC \cite{Bao1,Bao2009,Bao2013,Wz1,WangH} and also known as the imaginary
time method in the physics literature.
In detail, we evolve an initial state $\Phi_0:=(\phi_1^{(0)},\phi_2^{(0)})^T$ through the following GFDN
\begin{equation}\label{eq:gfdn}
\begin{split}
&\partial_t \phi_1=\left[\frac{1}{2}\nabla^2
-V_1(\bx)-ik_0\p_x-\frac{\delta}{2}
-\sum_{l=1}^2\beta_{1l}|\phi_l|^2\right]\phi_1-\frac{\Omega}{2}
\phi_2, \quad t\in[t_n, t_{n+1}),\\
&\partial _t \phi_2=\left[\frac{1}{2}\nabla^2
-V_2(\bx)+ik_0\p_x+\frac{\delta}{2}-\sum_{l=1}^2\beta_{2l}|\phi_l|^2\right]
\phi_2-\frac{\Omega }{2}
\phi_1,\quad t\in[t_{n},t_{n+1}),\\
&\phi_1(\bx,t_{n+1})=\frac{\phi_1(\bx,t_{n+1}^-)}{\|\Phi(\cdot,t_{n+1}^-)\|},
\quad \phi_2(\bx,t_{n+1})=\frac{\phi_2(\bx,t_{n+1}^-)}{\|\Phi(\cdot,t_{n+1}^-)\|},\quad \bx\in\mathbb R^d,\\
&\phi_1(\bx,0)=\phi_1^{(0)}(\bx),\quad \phi_2(\bx,0)=\phi_2^{(0)}(\bx),\quad \bx\in\mathbb R^d.
\end{split}
\end{equation}

Due to the confining potentials $V_1(\bx)$ and $V_2(\bx)$,
the ground state $\Phi_g(\bx)$ decays exponentially fast when $|\bx|\to\infty$,
thus in practical computations, the above GFDN \eqref{eq:gfdn}
is first truncated on a bounded large computational domain $U$, e.g.
an interval $[a,b]$ in 1D, a rectangle $[a,b]\times[c,d]$ in 2D and
a box $[a,b]\times[c,d]\times[e,f]$ in 3D, with periodic boundary conditions.
Then the GFDN on $U$ can be further discretized in space via the pseudospectral
method with the Fourier basis or second-order central finite difference method
and in time via backward Euler scheme \cite{Bao2013,Bao2006,Wz1}. For details, we refer to
\cite{Bao2009,Bao2013,Bao2006,Wz1} and references therein.

\begin{figure}[htb]
\centerline{
(a)\psfig{figure=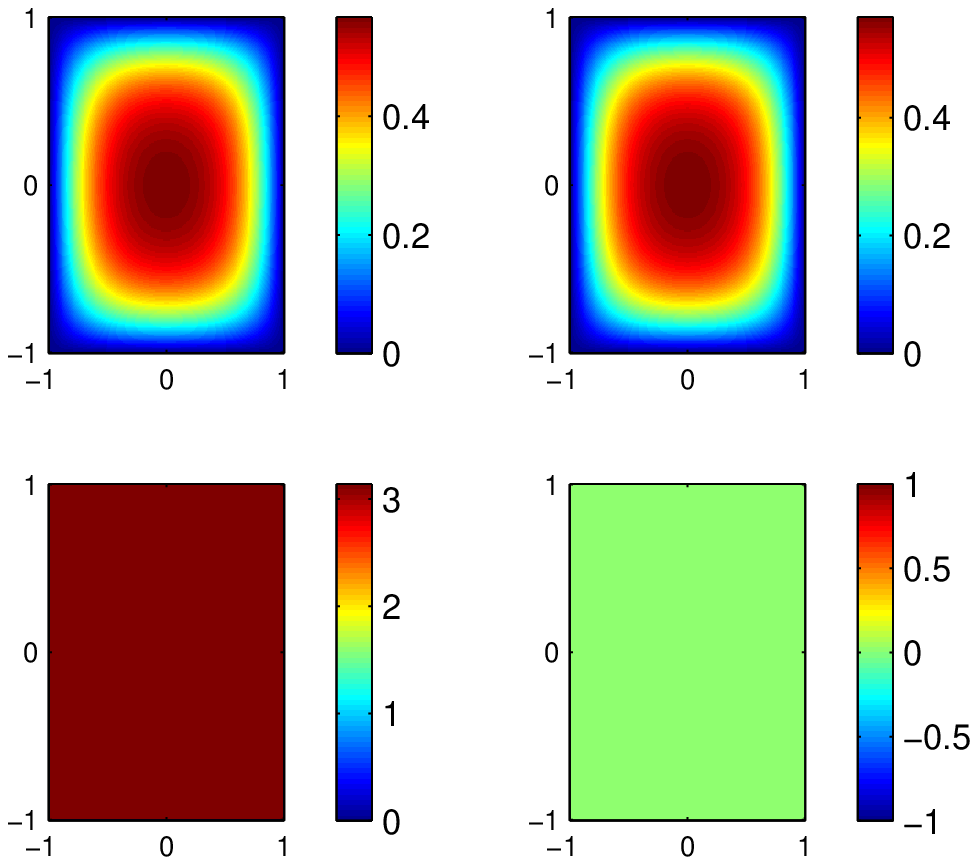,height=5cm,width=6cm,angle=0} \quad
(b)\psfig{figure=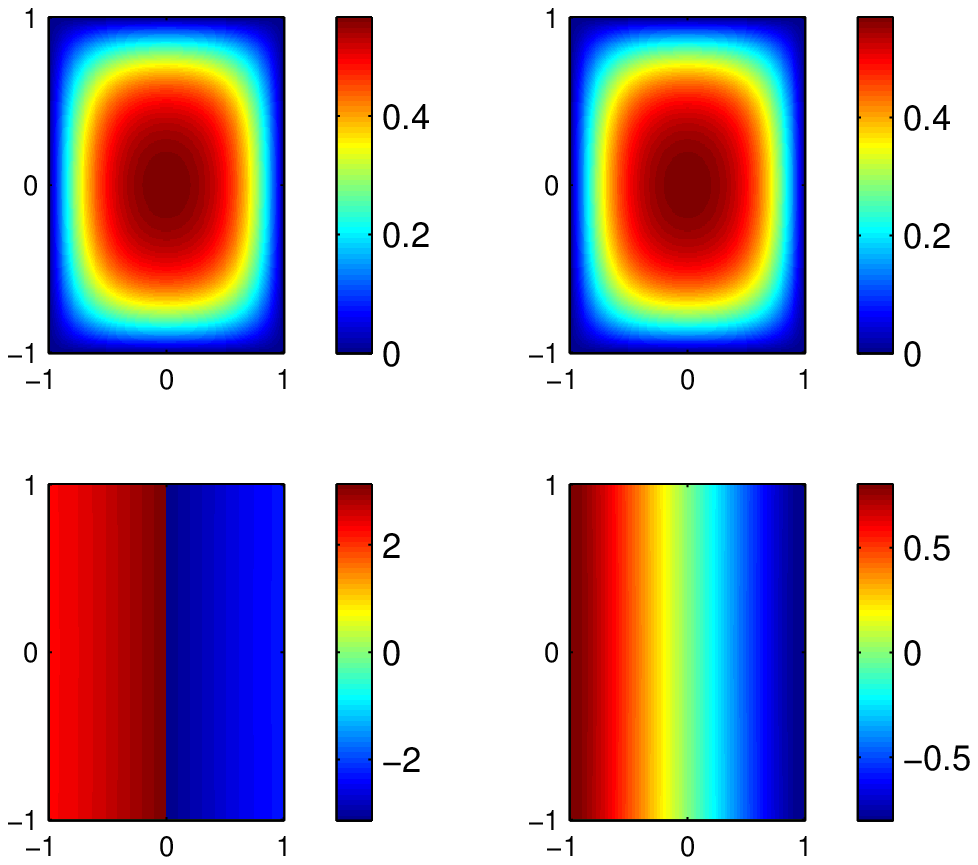,height=5cm,width=6cm,angle=0}}
\centerline{
(c)\psfig{figure=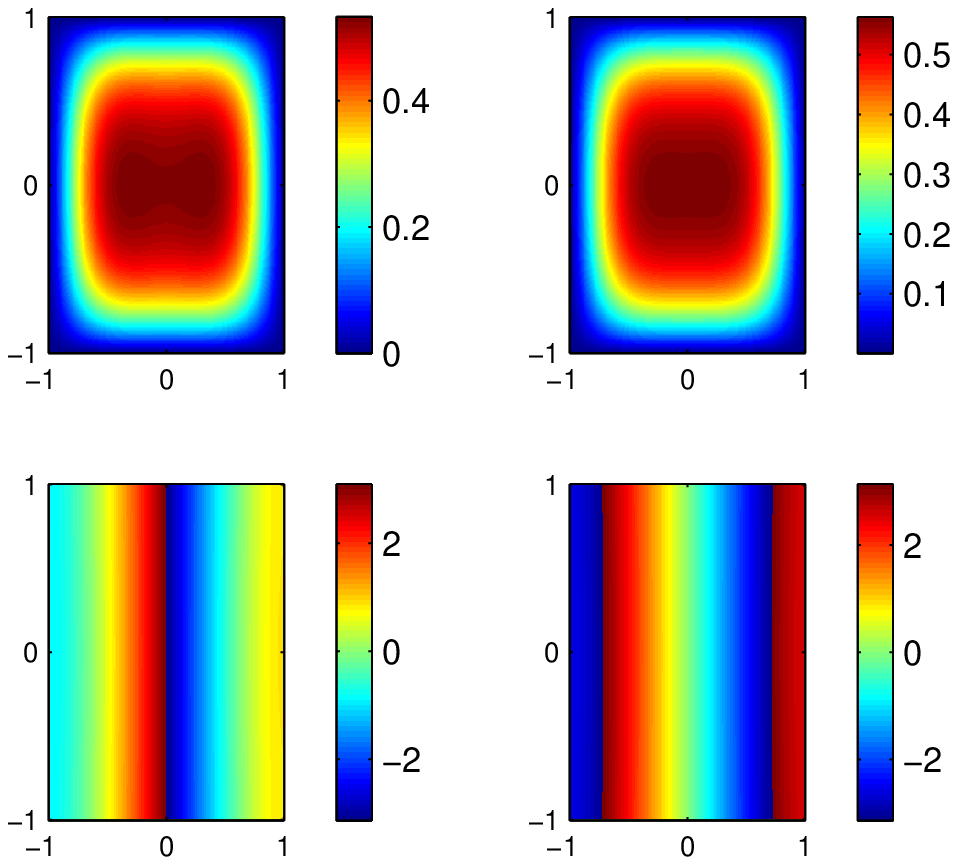,height=5cm,width=6cm,angle=0} \quad
(d)\psfig{figure=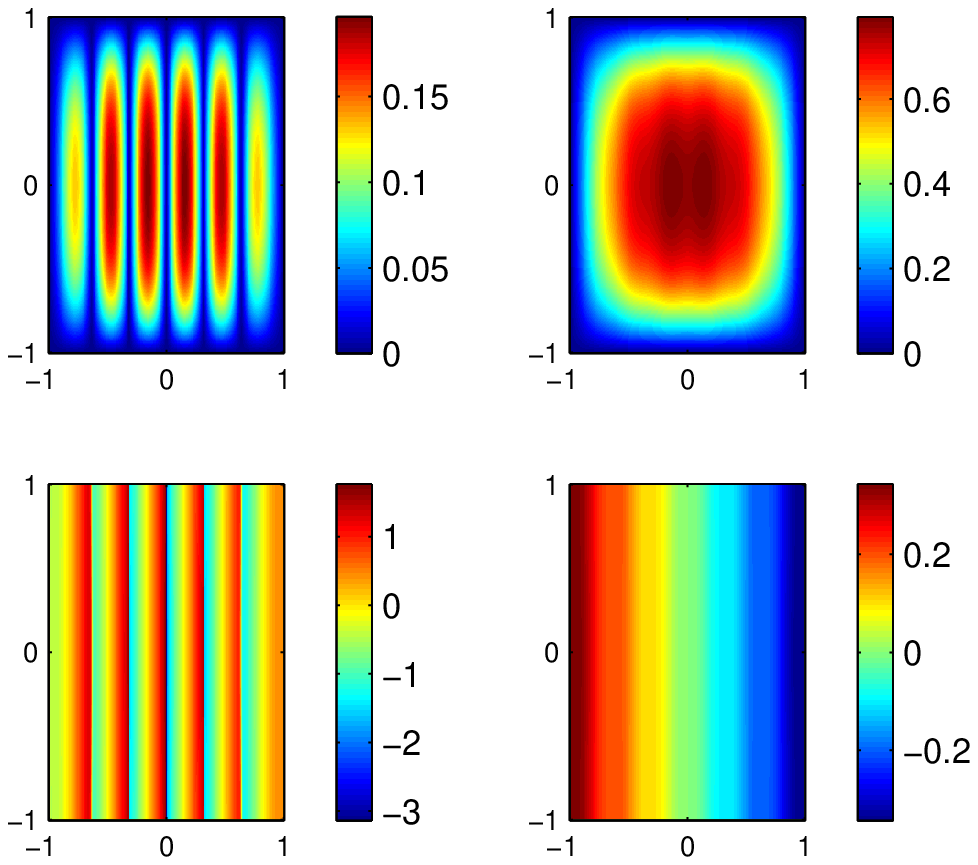,height=5cm,width=6cm,angle=0}}
\centerline{
(e)\psfig{figure=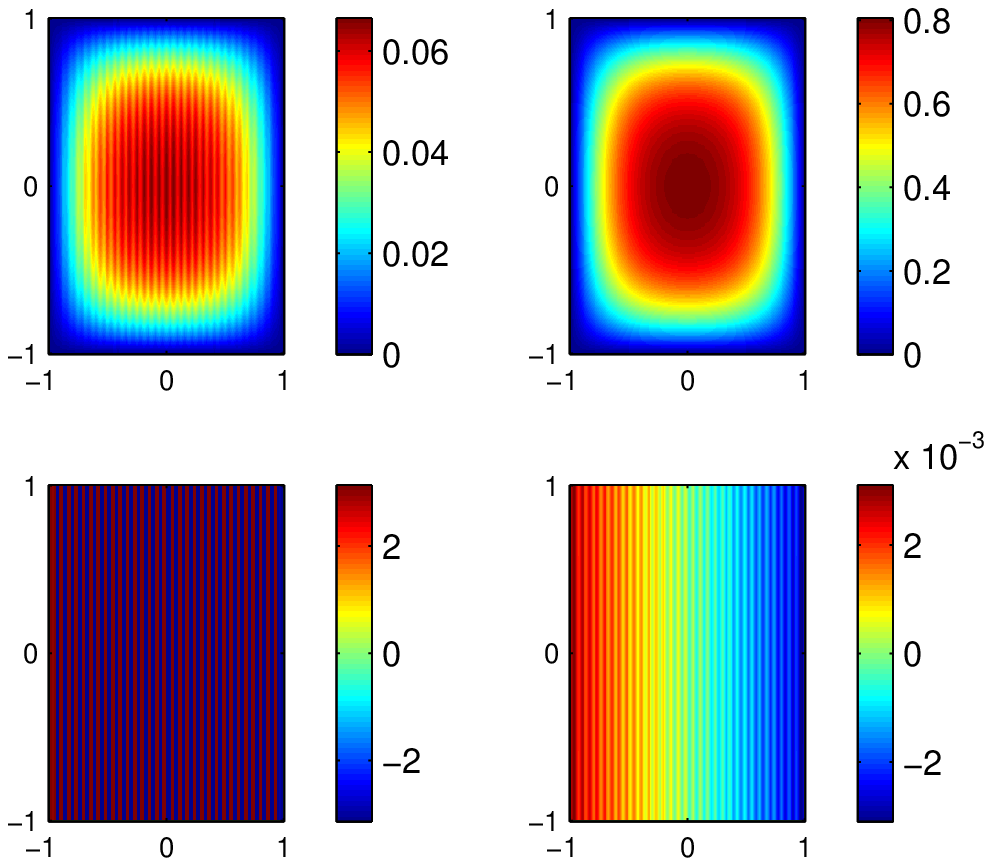,height=5cm,width=6cm,angle=0} \quad
(f)\psfig{figure=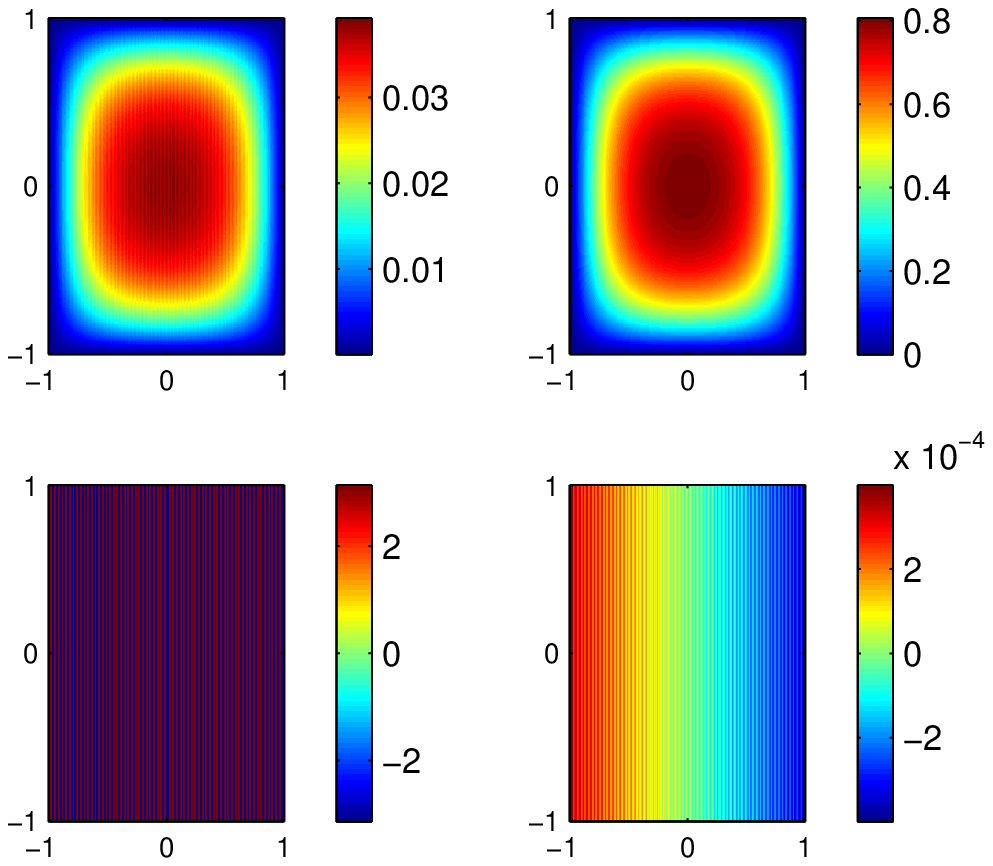,height=5cm,width=6cm,angle=0}}

\caption{Ground states $\tilde{\Phi}_g=(\tilde{\phi}_1^g,\tilde{\phi}_2^g)^T$
for a SO-coupled BEC in 2D with $\Omega=50$, $\delta=0$,
$\beta_{11}=10$, $\beta_{12}=\beta_{21}=\beta_{22}=9$ for:
(a) $k_0=0$, (b) $k_0=1$, (c) $k_0=5$, (d) $k_0=10$,
(e) $k_0=50$, and (f) $k_0=100$. In each subplot, top panel shows densities
and bottom panel shows phases of the ground state $\tilde{\phi}_1^g$ (left column)
and $\tilde{\phi}_2^g$ (right column). \label{fig:1}}
\end{figure}

\begin{figure}[htb]
\centerline{
(a)\psfig{figure=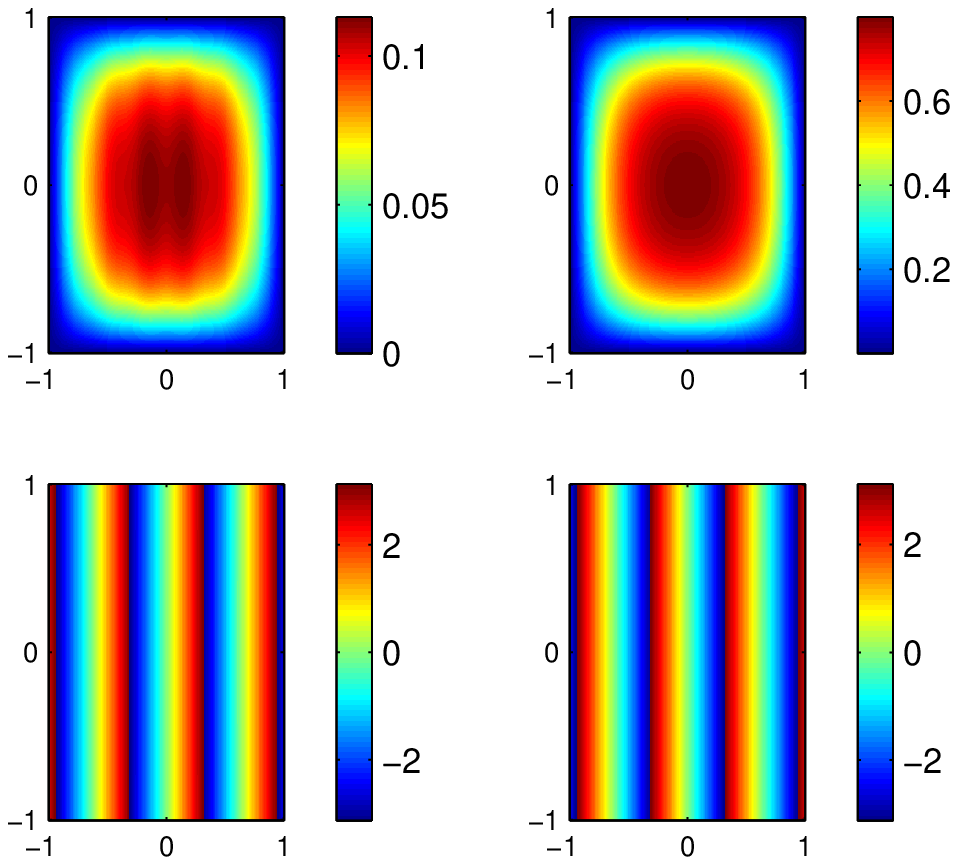,height=5cm,width=6cm,angle=0} \quad
(b)\psfig{figure=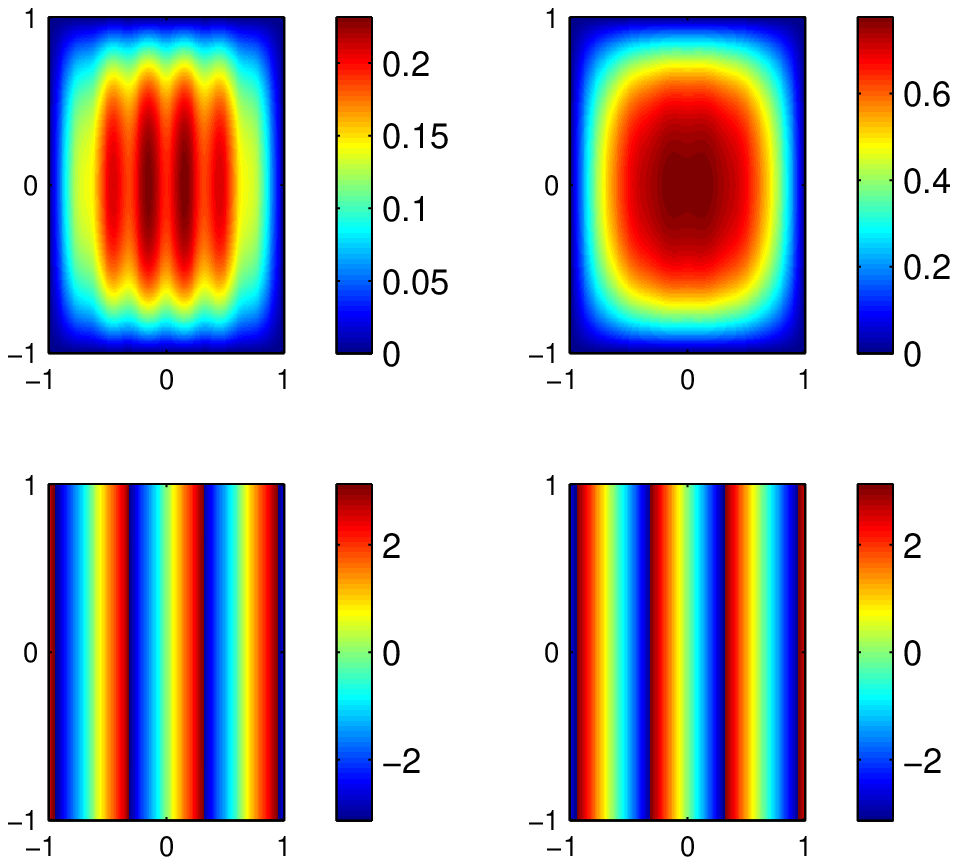,height=5cm,width=6cm,angle=0}}
\centerline{
(c)\psfig{figure=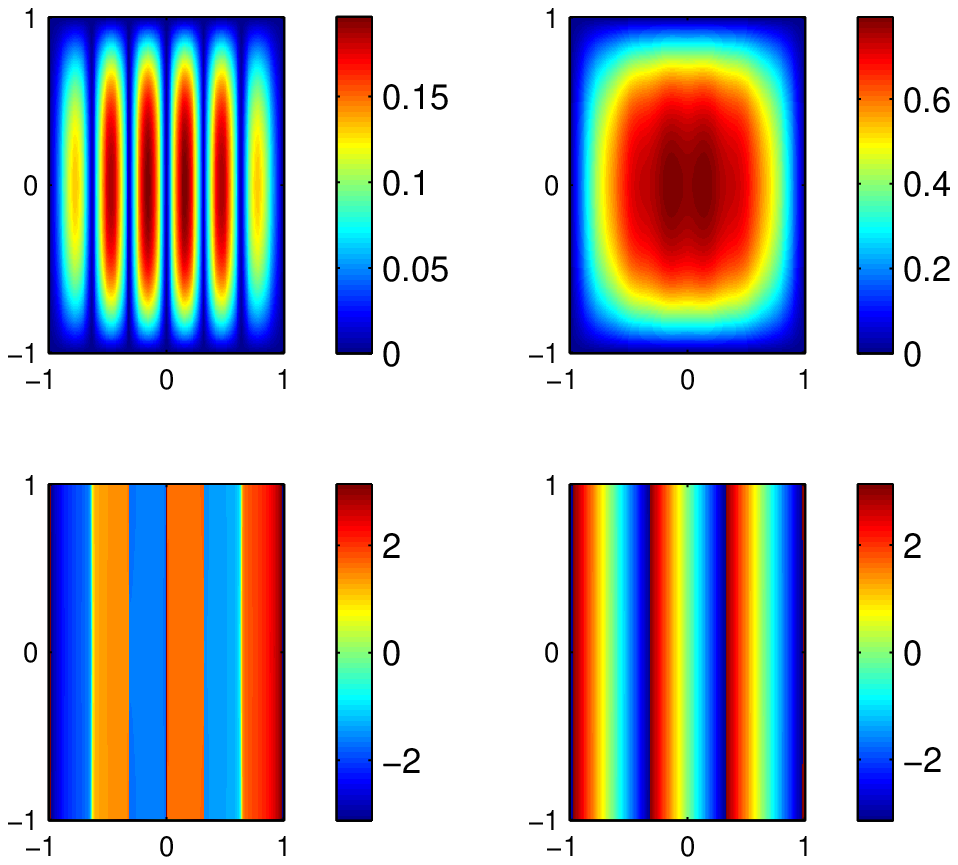,height=5cm,width=6cm,angle=0} \quad
(d)\psfig{figure=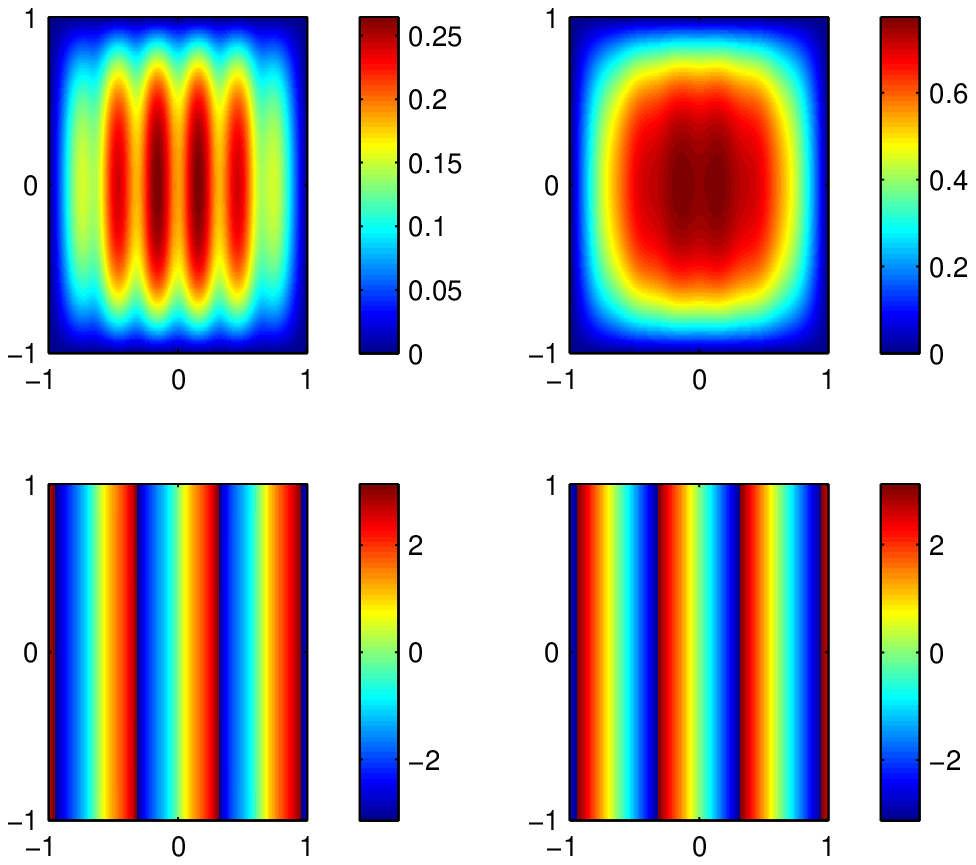,height=5cm,width=6cm,angle=0}}
\centerline{
(e)\psfig{figure=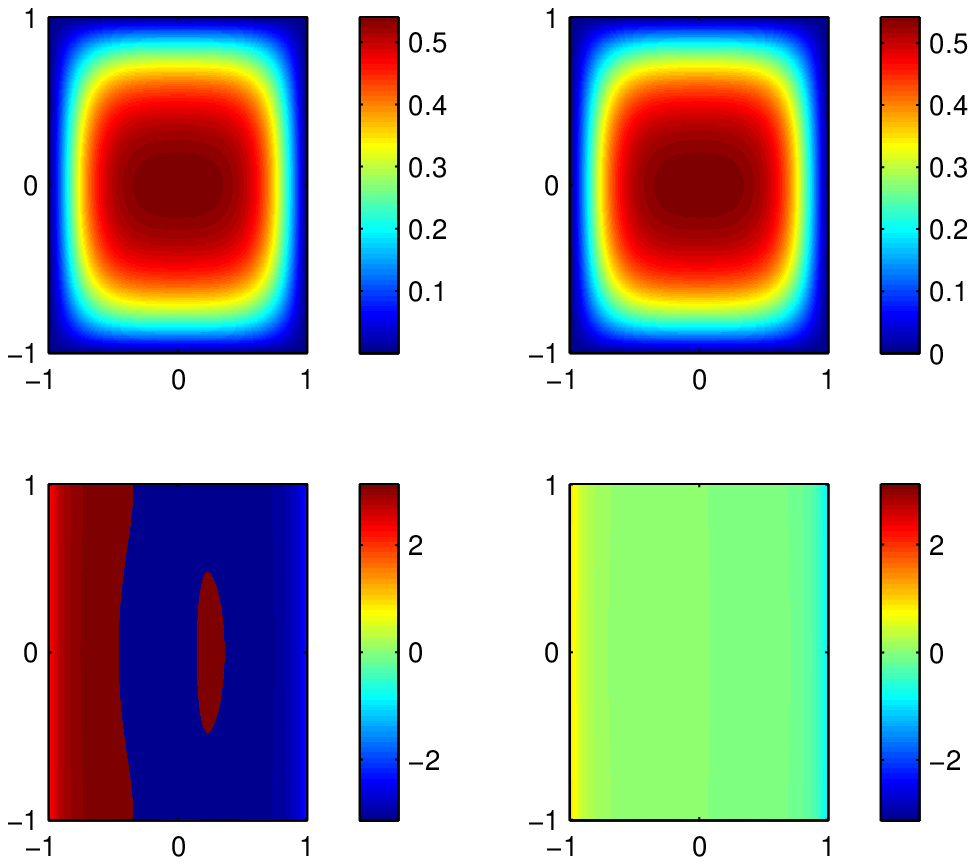,height=5cm,width=6cm,angle=0} \quad
(f)\psfig{figure=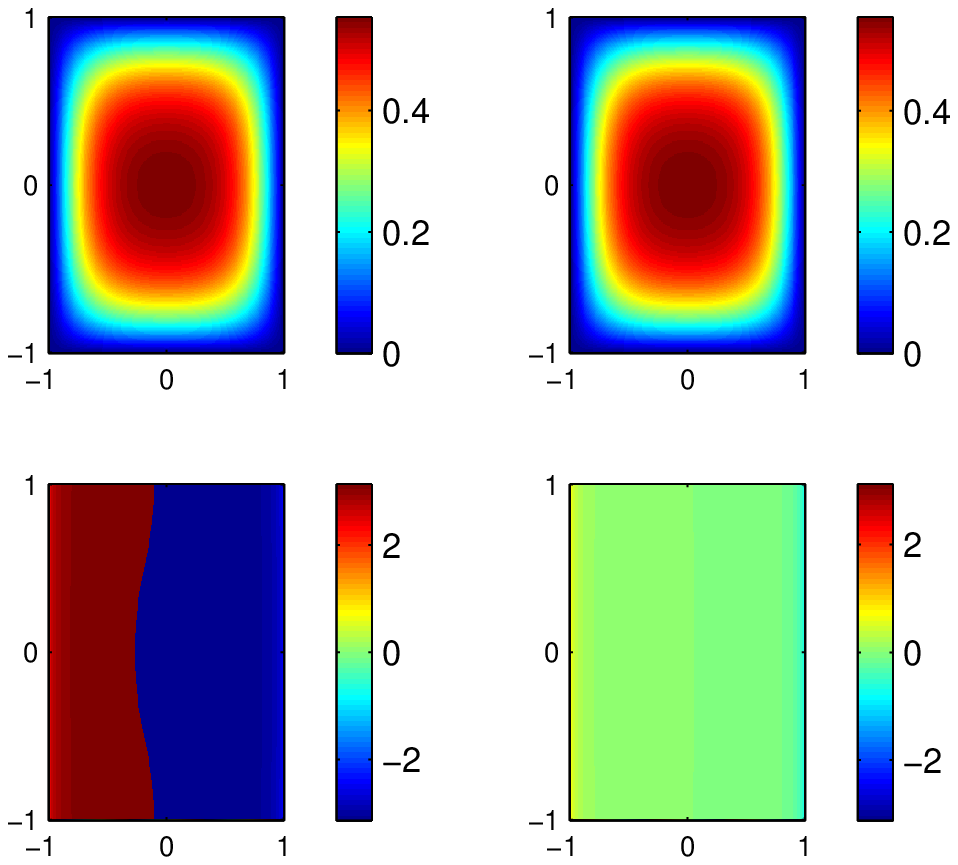,height=5cm,width=6cm,angle=0}}

\caption{\label{fig:2}Ground states $\Phi_g=(\phi_1^g,\phi_2^g)^T$
for a SO-coupled BEC in 2D with $k_0=10$, $\delta=0$,
$\beta_{11}=10$, $\beta_{12}=\beta_{21}=\beta_{22}=9$ for:
(a) $\Omega=1$, (b) $\Omega=10$, (c) $\Omega=50$, (d) $\Omega=200$,
(e) $\Omega=300$, and (f) $\Omega=500$. In each subplot, top panel shows densities
and bottom panel shows phases of the ground state $\phi_1^g$ (left column)
and $\phi_2^g$ (right column).}
\end{figure}

\begin{remark} If the box potential
\begin{equation}\label{eq:box}
V_{{\rm box}}(\bx)=\begin{cases}0,&\bx \in U,\\
+\infty,&\text{otherwise},
\end{cases}
\end{equation}
is used in the CGPEs \eqref{eq:cgpe1}
instead of the harmonic potentials (\ref{eq:potential:ho}),
due to the appearance of the SO coupling, in order to compute the ground state,
it is better to construct the GFDN based on (\ref{eq:minimize1}) and then discretize it
via the backward Euler sine pseudospectral (BESP) method due to that the homogeneous Dirichlet
boundary condition on $\partial U$ must be used in this case. Again, for details, we refer to
\cite{Bao2009,Bao2013,Bao2006,Wz1} and references therein.
\end{remark}

To test the efficiency and accuracy of the above numerical method for computing
the ground state of SO-coupled BECs, we take $d=2$,  $\delta=0$,
$\beta_{11}:\beta_{12}:\beta_{22}=1:0.9:0.9$ with $\beta_{11}=10$ in
\eqref{eq:cgpe1}. The potential $V_j(\bx)$ ($j=1,2$)
is taken as the box potential given in \eqref{eq:box} with $U=[-1,1]\times[-1,1]$.
We compute the ground state via the above BESP method with mesh size $h=\frac{1}{128}$
and time step $\tau=0.01$ ($\tau=0.001$ for large $\Omega$).
For the chosen parameters, it is easy to find that when $\Omega=0$, the ground
state $\Phi_g$ satisfies $\phi_1^g=0$ \cite{Bao1,Bao2009}.
Figure \ref{fig:1} shows the ground state $\tilde{\Phi}^g=(\tilde{\phi}_1^g,\tilde{\phi}_2^g)$
of \eqref{eq:minimize1} with $\Omega=50$ for different $k_0$,
which clearly demonstrates that as $k_0\to\infty$,  effect of $\Omega$
disappears. This is consistent with Theorem \ref{thm:k0change1}.
Figure \ref{fig:2} depicts the ground state $\Phi^g$ with $k_0=10$ for different $\Omega$,
from which we can observe that
$\phi_1^g$ and $\phi_2^g$ tend to have the same density profile with opposite
phase. This confirms Theorem \ref{thm:ogchange}.

\subsection{For computing dynamics}
In order to compute the dynamics of a SO-coupled BEC based on the CGPEs
\eqref{eq:cgpe1}, we usually truncate it onto a bounded computational domain $U$, e.g.
an interval $[a,b]$ in 1D,
a rectangle $[a,b]\times[c,d]$ in 2D and
a box $[a,b]\times[c,d]\times[e,f]$ in 3D, equipped
with periodic boundary conditions. Then
the CGPEs \eqref{eq:cgpe1} can be solve via a time-splitting technique to decouple
the nonlinearity \cite{Bao2003,Bao1,Bao2013,Antoine}. From $t_n$ to $t_{n+1}$, one first solves
\be\label{eq:split1}
\begin{split}
&i\p_t\psi_1=\left(-\frac12\Delta+ik_0\p_x+\frac{\delta}{2}\right)\psi_1+\frac{\Omega}{2}\psi_2,\\
&i\p_t\psi_2=\left(-\frac12\Delta-ik_0\p_x-\frac{\delta}{2}\right)\psi_2+\frac{\Omega}{2}\psi_1,
\end{split} \qquad \bx\in U,
\ee
for  time $\tau$, followed by solving
\be\label{eq:split2}
\begin{split}
&i\p_t\psi_1=\left(V_1(\bx)+\beta_{11}|\psi_1|^2+\beta_{12}|\psi_2|^2\right)\psi_1,\\
&i\p_t\psi_2=\left(V_2(\bx)+\beta_{21}|\psi_1|^2+\beta_{22}|\psi_2|^2\right)\psi_2,
\end{split} \qquad \bx\in U,
\ee
for another time $\tau$. Eq. \eqref{eq:split1} with periodic boundary conditions can be discretized
by the Fourier spectral method in space and then integrated in time {\sl exactly} \cite{Bao2003,Bao1,Bao2013,Antoine}.
Eq. \eqref{eq:split2} leaves the densities $|\psi_1|$ and $|\psi_2|$ unchanged and it can
be integrated in time {\sl exactly} \cite{Bao2003,Bao1,Bao2013,Antoine}. Then
a full discretization scheme can be constructed via a
combination of the splitting steps \eqref{eq:split1} and \eqref{eq:split2}
with a second-order or higher-order time-splitting methods \cite{Bao2003,Bao1,Bao2013,Antoine}.

For the convenience of the readers, here we present the method in 1D for the simplicity of notations.
Extensions to 2D and 3D are straightforward. In 1D, let $h=\Delta x=(b-a)/N$
($N$ an even positive integer), $x_j=a+jh$ ($j=0,\ldots,N$),
$\Psi_j^n=(\psi_{1,j}^n,\psi_{2,j}^n)^T$ be the numerical approximation
of $\Psi(x_j,t_n)=(\psi_1(x_j,t_n),\psi_{2}(x_j,t_n))^T$,
and for each fixed $l=1,2$, denote $\psi_{l}^n$  to be the vector consisting of
$\psi_{l,j}^n$ for $j\in{\cal T}_N=\{0,1,\ldots,N-1\}$.
From time $t=t_n$ to $t=t_{n+1}$,  a second-order time-splitting Fourier pseudospectral (TSFP)
method for the CGPEs \eqref{eq:cgpe1} in 1D reads \cite{Bao2003,Bao2013,Antoine}
\begin{equation}
\begin{split}
&\Psi^{(1)}_j
 =\frac{1}{N}\sum\limits_{k=-N/2}^{N/2-1}e^{i\mu_k(x_j-a)}\,Q_k^T\,e^{-\frac{i\tau}{4} U_k}\;Q_k(\widetilde{\Psi^n})_k,\\
&\Psi_j^{(2)}=e^{-i\tau P_j^{(1)}}\;\Psi_j^{(1)},  \qquad \qquad j=0,1,\ldots, N-1, \\
&\Psi^{n+1}_j
 =\frac{1}{N}\sum\limits_{k=-N/2}^{N/2-1}e^{i\mu_k(x_j-a)}\,Q_k^T\,e^{-\frac{i\tau}{4} U_k}\;Q_k(\widetilde{\Psi^{(2)}})_k,
\end{split}
 \end{equation}
where for each fixed $k=-\frac{N}{2},-\frac{N}{2}+1,\ldots,\frac{N}{2}-1$, $\mu_k=\frac{2 k\pi}{b-a}$,
$(\widetilde{\Psi}^n)_k=((\widetilde{\psi_{1}^n})_k,(\widetilde{\psi}_{2}^n)_k)^T$
with $(\widetilde{\psi_l^n})_k$ being the discrete Fourier transform
coefficients of $\psi_l^n$ ($l=1,2$), $U_k=\text{diag}\left(\mu_k^2+2\lambda_k, \mu_k^2-2\lambda_k\right)$
is a diagonal matrix, and
\[
 Q_k=\begin{pmatrix}\frac{\sqrt{\lambda_k-\chi_k}}{\sqrt{2\lambda_k}}&\frac{\frac{\Omega}{2}}
 {\sqrt{2\lambda_k(\lambda_k-\chi_k)}}\\
  \frac{-\sqrt{\lambda_k+\chi_k}}{\sqrt{2\lambda_k}}
&\frac{\frac{\Omega}{2}}{\sqrt{2\lambda_k(\lambda_k+\chi_k)}}\end{pmatrix}\quad \hbox{with}\quad
\chi_k=k_0\mu_k-\frac{\delta}{2},\quad \lambda_k=\frac{1}{2}\sqrt{4\chi_k^2+\Omega^2},
\]
and
$P_j^{(1)}=\text{diag}\left(V_1(x_j)+\sum\limits_{l=1}^2\beta_{1l}|\psi_{l,j}^{(1)}|^2,
V_2(x_j)+\sum\limits_{l=1}^2\beta_{2l}|\psi_{l,j}^{(1)}|^2\right)$ for $j=0,1\ldots,N-1$.

\subsection{Box potential case}
In some recent experiments of SO-coupled BEC, the box potential \eqref{eq:box} is used.
In this situation, due to that the homogeneous Dirichlet boundary condition on $\partial U$ must be used
for the CGPEs \eqref{eq:cgpe1},  similarly to the computation of the ground states,
it is better to adopt the CGPEs \eqref{eq:cgpe199:sec9} for computing the dynamics.
From $t=t_n$ to $t_{n+1}$, the CGPEs \eqref{eq:cgpe199:sec9} will be split into the
following three steps due to the appearance of the SO coupling.  One first solves
\be\label{eq:split1::2}
\begin{split}
&i\p_t\tilde{\psi}_1=\left(-\frac12\Delta+\frac{\delta}{2}\right)\tilde{\psi}_1,\\
&i\p_t\tilde{\psi}_2=\left(-\frac12\Delta-\frac{\delta}{2}\right)\tilde{\psi}_2,
\end{split}\qquad \bx\in U,
\ee
for  time step $\tau$,  then solves
\be\label{eq:split2:2}
\begin{split}
&i\p_t\tilde{\psi}_1=\left(V_1(\bx)+\beta_{11}|\tilde{\psi}_1|^2
+\beta_{12}|\tilde{\psi}_2|^2\right)\tilde{\psi}_1,\\
&i\p_t\tilde{\psi}_2=\left(V_2(\bx)+\beta_{12}|\tilde{\psi}_1|^2+\beta_{22}
 |\tilde{\psi}_2|^2\right)\tilde{\psi}_2,
\end{split}\qquad \bx\in U,
\ee
for time step $\tau$, followed by solving
\be\label{eq:split3:2}
\begin{split}
&i\p_t\tilde{\psi}_1=\frac{\Omega}{2} e^{-i2k_0x}\tilde{\psi}_2,\\
&i\p_t\tilde{\psi}_2=\frac{\Omega}{2} e^{i2k_0x}\tilde{\psi}_1,
\end{split}\qquad \bx\in U,
\ee
for  time step $\tau$. Again, Eq. \eqref{eq:split1::2} with homogeneous Dirichlet boundary
conditions can be discretized
by the sine spectral method in space and then integrated in time {\sl exactly} \cite{Bao2003,Bao1,Bao2013,Antoine}.
Eq. \eqref{eq:split2:2} leaves the densities $|\tilde\psi_1|$ and $|\tilde\psi_2|$ unchanged and it can
be integrated in time {\sl exactly} \cite{Bao2003,Bao1,Bao2013,Antoine}.
In addition, Eq. \eqref{eq:split3:2} is a linear ODE and can be integrated in time {\sl exactly} as
\be
\tilde{\Psi}(\bx,t_{n+1})=T(x)^\ast\,e^{-i\tau \Omega J}\, T(x)\, \tilde{\Psi}(\bx,t_n),\ \hbox{with}\ T(x)=\frac{1}{\sqrt{2}}\begin{pmatrix}1 &e^{-i2k_0x}\\
-1 &e^{-i2k_0x}\end{pmatrix},
\ee
where $J=\text{diag}(-1,1)$ and $T(x)^\ast=\overline{T(x)}^T$ is the adjoint matrix of $T(x)$.
Then a full discretization scheme can be constructed via a
combination of the splitting steps \eqref{eq:split1::2}-\eqref{eq:split3:2}
with a second-order  method \cite{Bao2003,Bao1,Bao2013,Antoine}. The details are omitted here for brevity.

\section{Dynamics of SO-coupled BEC}
In this section, we study dynamical properties,
in particular the motion of center-of-mass, of a SO-coupled BEC by using the CGPEs \eqref{eq:cgpe1}.
\subsection{Dynamics of center-of-mass}
Let $\Psi=(\psi_1,\psi_2)^T$ be the wave function describing the SO-coupled BEC, which  is governed by the CGPEs \eqref{eq:cgpe1}. Define
 the center-of-mass of the BEC as
\be \label{com321}
\bx_c(t)
=\int_{\mathbb R^d}\bx\sum\limits_{j=1}^2|\psi_j(\bx,t)|^2\,d\bx,\qquad t\ge0,
\ee
and the momentum as
\be\label{eq:momentum}
{\bP}(t)=\int_{\mathbb R^d}\sum\limits_{j=1}^2\text{Im}(\overline{\psi_j(\bx,t)}
\nabla\psi_j(\bx,t))\,d\bx,\qquad t\ge0,
\ee
where $\text{Im}(f)$ denotes the imaginary part of the function $f$.
In addition, we introduce the difference of the masses  $N_1(t)$ and $N_2(t)$ in
\eqref{eq:Njtt1} of the two components in the SO-coupled BEC as
\be\label{eq:Ndiff}
\delta_N(t):=N_1(t)-N_2(t)=\int_{\mathbb R^d}\left[|\psi_1(\bx,t)|^2-|\psi_2(\bx,t)|^2\right]\,d\bx,
\qquad t\ge0.
\ee
Then the following lemma holds.

\begin{lemma}\label{lem:com} Let $V_1(\bx)=V_2(\bx)$ be the
$d$-dimensional ($d=1,2,3$) harmonic potentials given in \eqref{eq:potential:ho},
then the  motion of the center-of-mass $\bx_c(t)$ for the CGPEs \eqref{eq:cgpe1} is governed by
        \be\label{eq:xc}
\ddot\bx_c(t)=-\Lambda \bx_c(t)-2k_0\Omega\,\text{Im}\left(\int_{\mathbb R^d}\overline{\psi_1(\bx,t)}\psi_2(\bx,t)\,d\bx\right)\,{\bf e}_x, \qquad t>0,
        \ee
where $\Lambda$ is a $d\times d$ diagonal matrix with $\Lambda=\gamma_x^2$ in 1D ($d=1$), $\Lambda=\text{diag}(\gamma_x^2,\gamma_y^2)$ in 2D ($d=2$)
and $\Lambda=\text{diag}(\gamma_x^2,\gamma_y^2,\gamma_z^2)$ in 3D ($d=3$),
${\bf e}_x$ is the unit vector for $x$-axis.
The  initial conditions for \eqref{eq:xc} are given as
        \begin{equation*}
        \bx_c(0)=\int_{\mathbb R^d}\bx\sum\limits_{j=1}^2|\psi_j(\bx,0)|^2\,d\bx,\qquad
        \dot{\bx}_c(0)=\bP(0)-k_0 \delta_N(0)\,{\bf e}_x.
        \end{equation*}
 In particular, \eqref{eq:xc} implies that the center-of-mass $\bx_c(t)$ is
 periodic  in $y$-component with frequency $\gamma_y$ when $d=2,3$, and in $z$-component
 with frequency $\gamma_z$ when $d=3$.
 If $k_0\Omega=0$, $\bx_c(t)$ is also periodic in $x$-component with frequency $\gamma_x$.
\end{lemma}
\begin{proof} For $j=1,2$, differentiating $\bx_j(t)=\int_{\mathbb R^d}\bx|\psi_j(\bx,t)|^2\,d\bx$,
using the CGPEs \eqref{eq:cgpe1}
and integral by parts, we find
\begin{equation*}
\dot \bx_j(t)=\bP_j(t)-k_0(3-2j)N_j(t)\,{\bf e}_x-\frac{i\Omega}{2}\int_{\mathbb R^d}\bx\left(\overline{\psi}_j\psi_{3-j}-\psi_j\overline{\psi}_{3-j}\right)\,d\bx,
\end{equation*}
where $\bP_j(t):=\int_{\mathbb R^d}\text{Im}(\overline{\psi_j(\bx,t)}\nabla\psi_j(\bx,t))\,d\bx$.
Summing the above equation for $j=1,2$ and noticing (\ref{com321}) and (\ref{eq:momentum}), we get
\be\label{eq:xdotc}
\dot \bx_c(t)=\bP(t)-k_0\delta_N(t)\, {\bf e}_x,\qquad t\ge0.
\ee
Differentiating \eqref{eq:xdotc} once more, we get
\be\label{eq:xddotc}
\ddot \bx_c(t)=\dot\bP(t)-k_0\dot \delta_N(t)\,{\bf e}_x.
\ee
We now compute the RHS of \eqref{eq:xddotc}. Firstly, for $j=1,2$,
differentiating $\bP_j(t)$,  making use of the CGPEs \eqref{eq:cgpe1} and integral by parts,
we get
\begin{equation*}
\dot\bP_j(t)=
\int_{\mathbb R^d}\left[-|\psi_j|^2\nabla V_j(\bx)-\beta_{12}|\psi_{3-j}|^2\nabla|\psi_j|^2+\Omega\,\text{Re}(\overline{\psi}_{3-j}\nabla\psi_j)\right]\,d\bx,
\end{equation*}
which immediately gives
\be\label{eq:Pdot}
\dot\bP(t)=-\int_{\mathbb R^d}\Lambda \bx \sum\limits_{j=1}^2|\psi_j|^2\,d\bx=-\Lambda \bx_c(t),
\ee
with $\Lambda$ being the diagonal matrix described in the lemma. Secondly, for $j=1,2$,
differentiating $N_j(t)$,
making use of the CGPEs \eqref{eq:cgpe1} and integral by parts, we obtain
 \begin{equation*}
\dot N_j(t)=-\frac{i\Omega}{2}\int_{\mathbb R^d}\left(\overline{\psi}_j\psi_{3-j}-\overline{\psi}_{3-j}\psi_j\right)\,d\bx,
  \end{equation*}
  which again immediately implies
  \be\label{eq:DNdt}
  \dot \delta_N(t)=2\Omega\, \text{Im}\int_{\mathbb R^d}\overline{\psi_1(\bx,t)}\psi_2(\bx,t)\,d\bx.
  \ee
 Combining \eqref{eq:DNdt}, \eqref{eq:Pdot} and \eqref{eq:xddotc}, we draw the conclusion.
\end{proof}

From Lemma \ref{lem:com}, the effect of SO coupling on the motion of the center-of-mass
$\bx_c(t)$  appears in the $x$-component. Denote the $x$-component of $\bx_c(t)$
as $x_c(t)$, and the $x$-component of $\bP(t)$ as $P^x(t)$. Then we have the following results:

\begin{theorem}\label{thm:comapp} Let $V_1(\bx)=V_2(\bx)$ be the harmonic potential as \eqref{eq:potential:ho} in $d$ dimensions ($d=1,2,3$) and $k_0\Omega\neq0$. For the $x$-component $x_c(t)$ of the
center-of-mass $\bx_c(t)$ of the  CGPEs \eqref{eq:cgpe1} with any initial data $\Psi(\bx,0):=\Psi_0(\bx)$ satisfying
$\|\Psi_0\|=1$, we have
\be\label{eq:xcc}
x_c(t)=x_0\cos(\gamma_xt)+\frac{P_0^x}{\gamma_x}\sin(\gamma_xt)-k_0\int_0^t\cos(\gamma_x(t-s))\delta_N(s)\,ds,
\qquad t\ge0,
\ee
where $x_0=\int_{\mathbb R^d}x\sum_{j=1}^2|\psi_j(\bx,0)|^2\,d\bx$ and $P_0^x=\int_{\mathbb R^d}\sum_{j=1}^2\text{Im}(\overline{\psi_j(\bx,0)}\p_x\psi_j(\bx,0))\,d\bx$.
In addition, if $\delta=0$, $\beta_{11}=\beta_{12}=\beta_{22}$ and $|k_0|$ is small, we can approximate the solution $x_c(t)$ as follows:

(i) If $|\Omega|=\gamma_x$,  we can get
\begin{equation*}
x_c(t)\approx\left(x_0-\frac{k_0}{2}\delta_N(0)t\right)\cos(\gamma_xt)+\frac{1}{\gamma_x}
\left(P_0^x-\frac{k_0}{2}\delta_N(0)-
{\rm sgn}(\Omega)\frac{\gamma_xk_0C_0}{2}t\right)\sin(\gamma_xt),
\end{equation*}
where  $C_0=2\text{Im}\int_{\mathbb R^d}\overline{\psi_1(\bx,0)}\psi_2(\bx,0)\,d\bx$.

(ii) If $|\Omega|\neq \gamma_x$,  we can get
\begin{equation*}
\begin{split}
x_c(t)\approx& \left(x_0+\frac{k_0C_0}{\gamma_x^2-\Omega^2}\right)\cos(\gamma_xt)
+\frac{1}{\gamma_x}\left(P_0^x-\frac{\gamma_x^2k_0\delta_N(0)}{\gamma_x^2-\Omega^2}\right)\sin(\gamma_xt)\\
&-\frac{k_0C_0}{\gamma_x^2-\Omega^2}\cos(\Omega t)+\frac{k_0 \delta_N(0)\Omega}{\gamma_x^2-\Omega^2} \sin(\Omega t).
\end{split}
\end{equation*}

Based on the above approximation,  if $|\Omega|=\gamma_x$ or $\frac{\Omega}{\gamma_x}$ is an
irrational number, $x_c(t)$ is not periodic; and if
$\frac{\Omega}{\gamma_x}$ is a rational number, $x_c(t)$ is a periodic function, but its frequency is different
from the trapping frequency $\gamma_x$.
\end{theorem}
\begin{proof} Solving \eqref{eq:xc} by the variation-of-constant formula and using \eqref{eq:DNdt}, we have
\begin{equation*}
x_c(t)=x_c(0)\cos(\gamma_xt)+\frac{P^x(0)-k_0\delta_N(0)}{\gamma_x}
\sin(\gamma_xt)-\frac{k_0}{\gamma_x}\int_0^t\cos(\gamma_x(t-s))\dot \delta_N(s)\,ds,
\end{equation*}
and \eqref{eq:xcc} follows by applying integration by parts.

In order to obtain the prescribed approximation, we first find the equation for $\delta_N(t)$. Differentiating \eqref{eq:DNdt} and using \eqref{eq:cgpe1}, we get
\begin{align*}
\ddot\delta_N(t)
&=-\Omega^2\delta_N(t)+2\Omega\,\text{Re}\int_{\mathbb R^d}\bigg[\big(V_1(\bx)-V_2(\bx)+\delta+(\beta_{11}-\beta_{21})|\psi_1|^2\\
&\qquad\qquad\qquad+(\beta_{12}-\beta_{22})|\psi_2|^2\big)\overline{\psi_1}\psi_2
+ik_0(\overline{\psi}_1\p_x\psi_2-\overline{\p_x\psi_1}\psi_2)\bigg]\,d\bx.
\end{align*}
Thus, if $|k_0|\ll1$ and $\delta=0$, $\beta_{11}=\beta_{12}=\beta_{22}$, the above equation is approximated by
\be
\ddot\delta_N(t)\approx -\Omega^2N_{\Delta}(t),
\ee
and the initial condition $\dot\delta_N(0)$ can be obtained via \eqref{eq:DNdt} with $t=0$.
Solving the above ODE, we find
\be\label{eq:Napp}
\delta_N(t)\approx \delta_N(0)\cos(\Omega t)+\frac{\dot\delta_N(0)}{\Omega}\sin(\Omega t).
\ee
Plugging \eqref{eq:Napp} into \eqref{eq:xcc}, we obtain the approximate solution of $x_c(t)$.
\end{proof}

\begin{figure}[htb]
\centerline{
(a)\psfig{figure=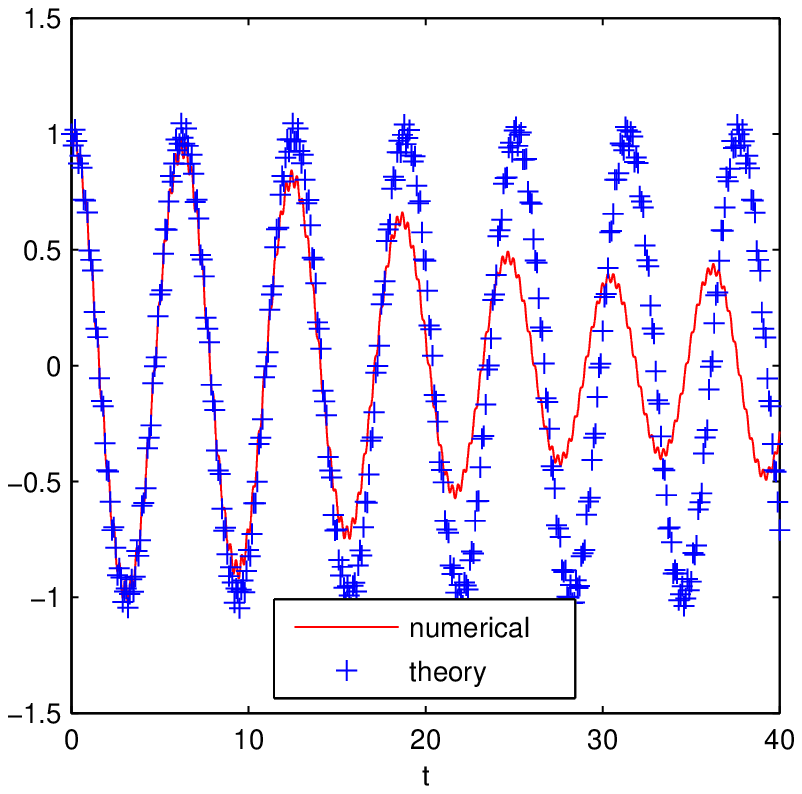,height=5cm,width=6cm,angle=0} \quad
(b)\psfig{figure=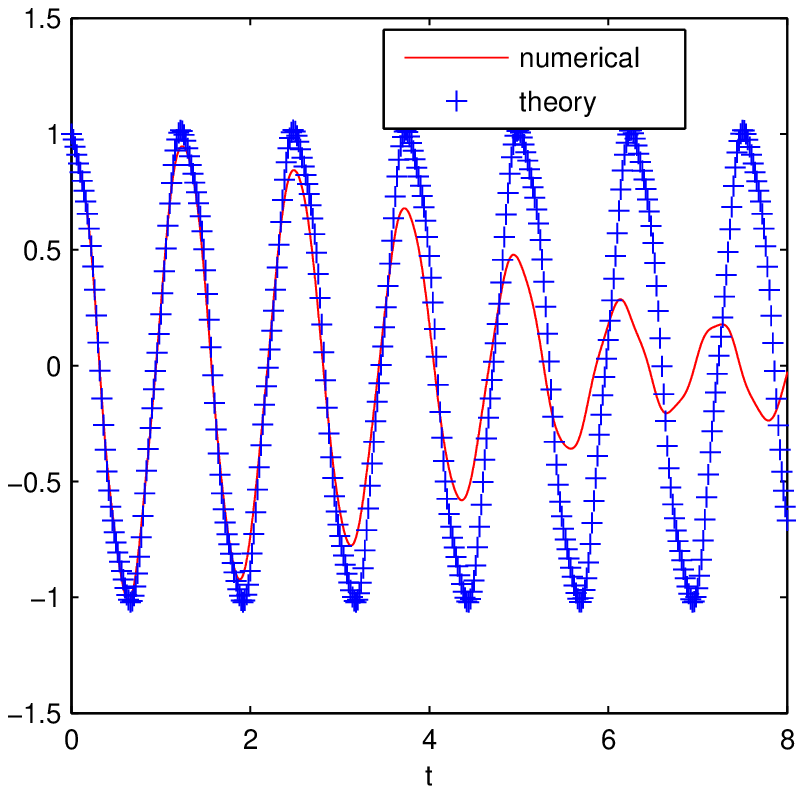,height=5cm,width=6cm,angle=0}}
\centerline{
(c)\psfig{figure=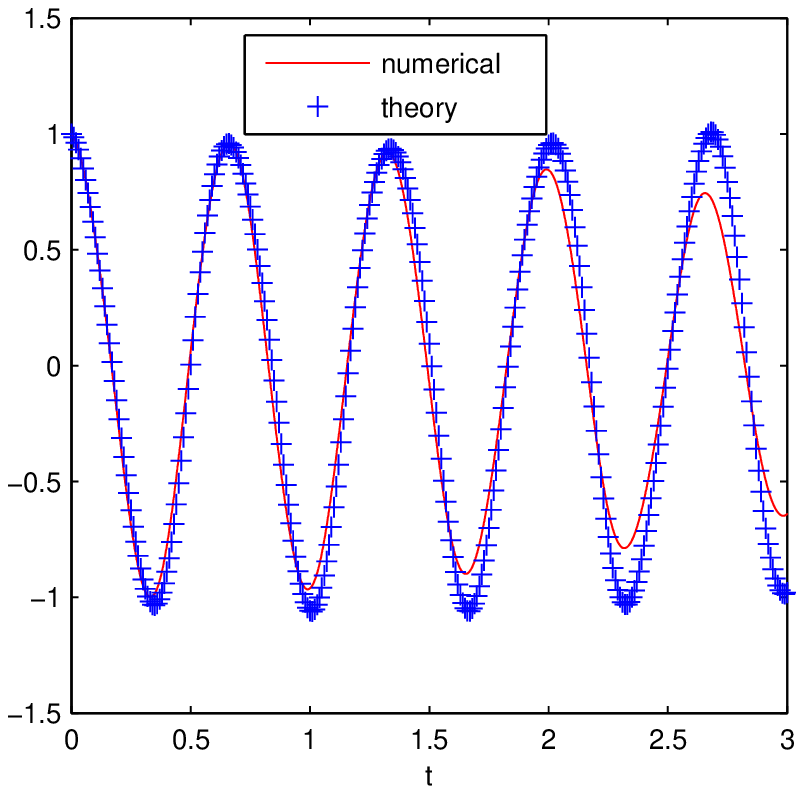,height=5cm,width=6cm,angle=0} \quad
(d)\psfig{figure=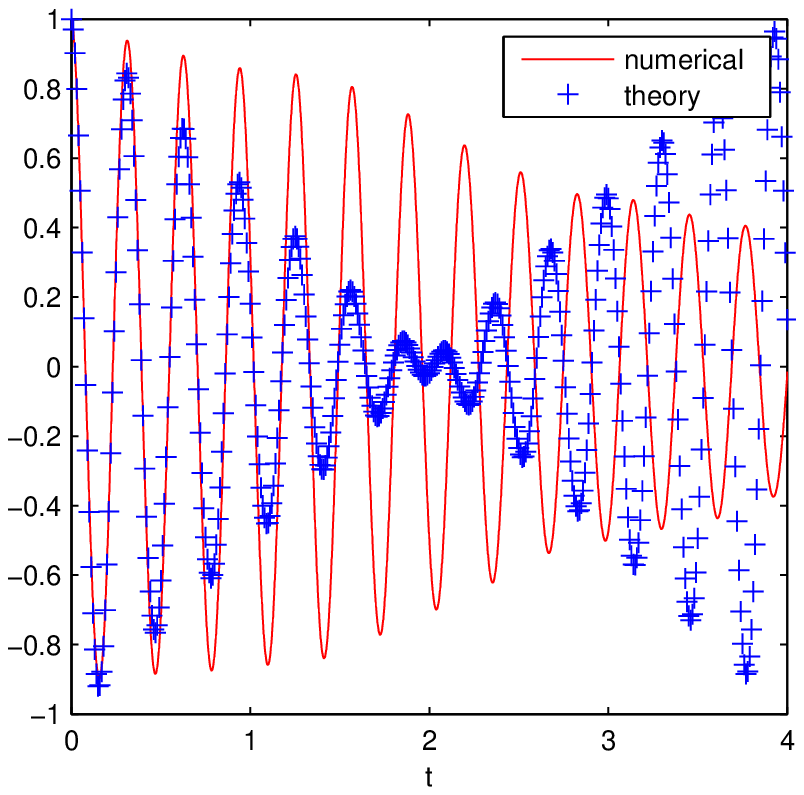,height=5cm,width=6cm,angle=0}}

\caption{Time evolution of the center-of-mass $x_c(t)$ for the CGPEs \eqref{eq:cgpe1}
obtained numerically from its numerical solution (i.e. labeled as 'numerical' with solid lines) and
asymptotically as in Theorem \ref{thm:comapp} (i.e. labeled as 'theory' with `+ + +')
with $\Omega=20$ and  $k_0=1$ for different $\gamma_x$: (a) $\gamma_x=1$, (b) $\gamma_x=5$, (c) $\gamma_x=3\pi$,
 and (d) $\gamma_x=20$.  \label{fig:0}}
\end{figure}

To verify the asymptotic (or approximate) results for $x_c(t)$ in Theorem \ref{thm:comapp},
we numerically solve  the CGPEs \eqref{eq:cgpe1} with (\ref{eq:potential:ho}) in  2D (i.e. $d=2$), take
$\beta_{11}=\beta_{12}=\beta_{22}=10$, $\delta=0$ and choose the initial data as
\be
\psi_1(\bx,0)=\pi^{-1/2}e^{-\frac{|\bx-\bx_0|^2}{2}},\quad \psi_2(\bx,0)=0,\quad \bx\in{\Bbb R}^2,
\ee
where $\bx_0=(1,1)^T$.
Figure \ref{fig:0} depicts time evolution of $x_c(t)$ obtained numerically and asymptotically as in Theorem \ref{thm:comapp}
with $\Omega=20$ and $k_0=1$ for different $\gamma_x$.
From this figure, we see that: for short time $t$, the approximation given in
Theorem \ref{thm:comapp} is very accurate; and when $t\gg1$, it becomes inaccurate, which
is due to that the assumption on  $\delta_N(t)$ obeying \eqref{eq:Napp}
becomes inaccurate.

\medskip

In fact, Theorem \ref{thm:comapp} is valid for any given initial data.
Now, we consider a kind of special initial data, i.e. shift of the ground state $\Phi_g=(\phi_1^g,\phi_2^g)^T$
 of \eqref{eq:minimize} for the CGPEs \eqref{eq:cgpe1}, i.e., the initial condition
 for \eqref{eq:cgpe1} is chosen as
 \be\label{eq:inishift}
 \psi_1(\bx,0)=\phi_1^g(\bx-\bx_0),\quad \psi_2(\bx,0)=\phi_2^g(\bx-\bx_0), \qquad
 \bx\in {\Bbb R}^d,
 \ee
where $\bx_0=x_0$ in 1D, $\bx_0=(x_0,y_0)^T$ in 2D and $\bx_0=(x_0,y_0,z_0)^T$ in 3D.
Then we have the approximate dynamical law for the center-of-mass in  $x$-direction $x_c(t)$.

\begin{theorem} \label{thm:com2}Suppose $V_1(\bx)=V_2(\bx)$ for $\bx\in{\Bbb R}^d$ are harmonic potentials given in \eqref{eq:potential:ho},
$\beta_{11}=\beta_{12}=\beta_{22}=\beta$ and the initial data for the CGPEs \eqref{eq:cgpe1} is taken as \eqref{eq:inishift}.
Using the local density approximation (LDA), the dynamics of the center-of-mass $x_c(t)$ can
be approximated by the following ODE
\be\label{eq:odecomgs}
\dot x_c(t)=P^x(t)-\frac{k_0[2k_0P^x(t)-\delta]}{\sqrt{[2k_0P^x(t)-
\delta]^2+\Omega^2}},\quad
\dot P^x(t)=-\gamma_x^2x_c(t),\qquad t\ge0,
\ee
with $x_c(0)=x_0$ and $P^x(0)=k_0\delta_N(0)$. In particular, the solution to \eqref{eq:odecomgs} is periodic,
and, in general, its frequency is different with the trapping frequency $\gamma_x$.
\end{theorem}

\begin{proof} The initial condition for the ODE \eqref{eq:odecomgs} comes
from the initial value \eqref{eq:inishift} for the  CGPEs \eqref{eq:cgpe1}.
We use LDA here, which means we will treat the BEC system as a uniform system $V_1(\bx)=V_2(\bx)=\text{constant}$ locally.
We begin with the uniform case. The evolution of the wave function
$\Psi=(\psi_1,\psi_2)^T$ is assumed to remain in the ground mode of
 the Hamiltonian
\be\label{eq:H}
{\bf H}=\begin{pmatrix}-\frac{\nabla^2}{2}+ik_0\p_x+\frac{\delta}{2}+\beta|\Psi|^2&\frac{\Omega}{2}\\
\frac{\Omega}{2}&-\frac{\nabla^2}{2}-ik_0\p_x-\frac{\delta}{2}+\beta|\Psi|^2\end{pmatrix},
\ee
and be localized near the center-of-mass $\bx_c(t)$ in physical space
and near the momentum $\bP(t)$ in the phase space.
Thus, the wave function can be written as
\be\label{eq:planewave}
\Psi=(\psi_1,\psi_2)^T=e^{i{\xi\cdot(\bx-\bx_c(t))}}\vec{v}, \quad \vec{v}\quad \text{is a  vector in}\quad \mathbb C^2,
\ee
and $ \xi=(\xi_1,\ldots,\xi_d)^T\in \mathbb R^d$ is centered around $\bP(t)$. Plugging \eqref{eq:planewave} into \eqref{eq:H}, we obtain a two-by-two matrix, and the
 two eigenvalues and the corresponding eigenvectors are
\be
{\cal E}_{\pm}=\frac{|\xi|^2}{2}+\beta|\vec{v}|^2\pm\tilde{\lambda},\quad
\vec{v}_{\pm}=\left(\frac{(\tilde{\lambda}\mp\tilde{\chi})^{1/2}}{(2\tilde{\lambda})^{1/2}},
\frac{\Omega}{2(2\tilde{\lambda}(\tilde{\lambda}\mp\tilde{\chi}))^{1/2}}\right)^T,
\ee
with $\tilde{\lambda}=\frac{1}{2}\sqrt{(2k_0\xi_1-\delta)^2+\Omega^2}$ and
 $\tilde{\chi}=k_0\xi_1-\frac{\delta}{2}$.
By our assumption that the evolution is in the lower eigenstate, we find
 $\vec{v}=|\vec{v}|\vec{v}_{-}$ and
\be
\frac{|\psi_1|^2}{|\psi_2|^2}=\frac{4(\tilde{\lambda}+\tilde{\chi})^2}{\Omega^2}.
\ee
Since the phase space is assumed to be localized around $\bP(t)$, we can approximate the
above equation by letting $\xi_1=P^x:=P^x(t)$ and we get
\be\label{eq:frac}
\frac{|\psi_1|^2}{|\psi_2|^2}\approx\frac{4(\lambda+\chi)^2}{\Omega^2},\quad
\lambda=\frac{1}{2}\sqrt{(2k_0P^x-\delta)^2+\Omega^2},\quad\chi=k_0P^x-\frac{\delta}{2}.
\ee
For the case with harmonic potentials $V_1(\bx)=V_2(\bx)$, we use LDA,
and we could get
the same relation between densities as \eqref{eq:frac} for each position $\bx$ which leads to
\be\label{eq:napprx}
\delta_N(t)=\frac{4(\lambda+\chi)^2-\Omega^2}{4(\lambda+\chi)^2+\Omega^2}.
\ee
Plugging \eqref{eq:napprx} into \eqref{eq:xdotc}, noticing \eqref{eq:Pdot},
we obtain the ODE system (\ref{eq:odecomgs}) approximating the dynamics of  $x_c(t)$. Using the equation \eqref{eq:odecomgs},
it is easy to find that
\be
\frac{d}{dt}\left(\gamma_x^2x_c^2(t)+(P^x(t))^2-
\sqrt{[2k_0P^x(t)-\delta]^2+\Omega^2}\right)
=0,
\ee
which shows $(x_c(t),P^x(t))^T$ is a closed curve and it is periodic.
\end{proof}

Again, to verify the asymptotic (or approximate) results for $x_c(t)$ in Theorem \ref{thm:com2},
we numerically solve  the CGPEs \eqref{eq:cgpe1} with (\ref{eq:potential:ho}) in  2D (i.e. $d=2$), take
$\beta_{11}=\beta_{12}=\beta_{22}=10$ and $\gamma_x=\gamma_y=2$, and choose the initial data as
\eqref{eq:inishift} with $\bx_0=(2,2)^T$ and the ground state computed numerically.
Figure \ref{fig:3} depicts time evolution of $x_c(t)$ obtained numerically and asymptotically as in Theorem
\ref{thm:com2} with different $\Omega$, $k_0$ and $\delta$.

\begin{figure}[htb]
\centerline{
(a)\psfig{figure=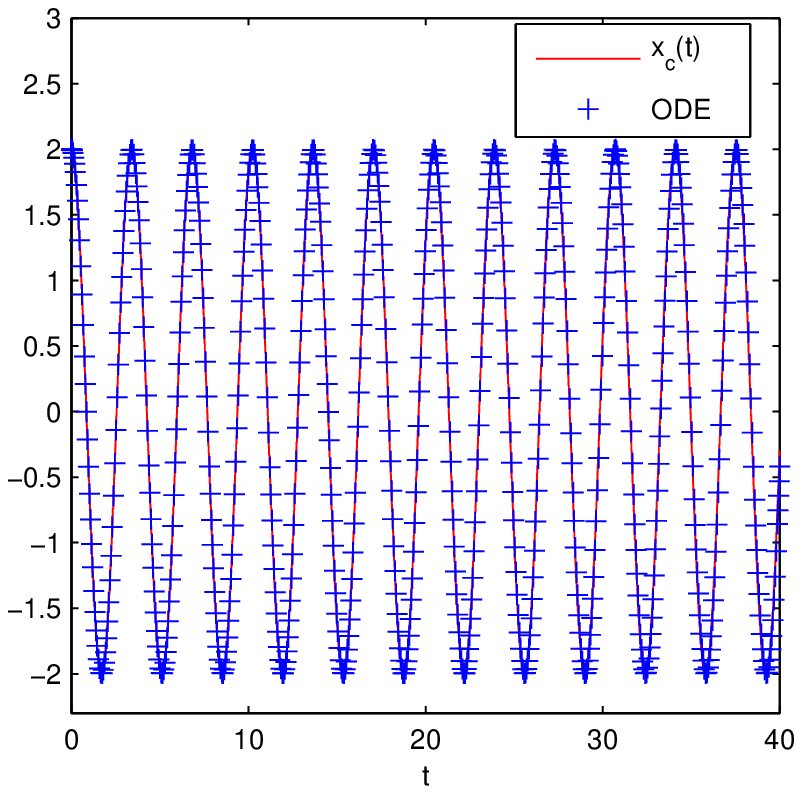,height=5cm,width=6cm,angle=0} \quad
(b)\psfig{figure=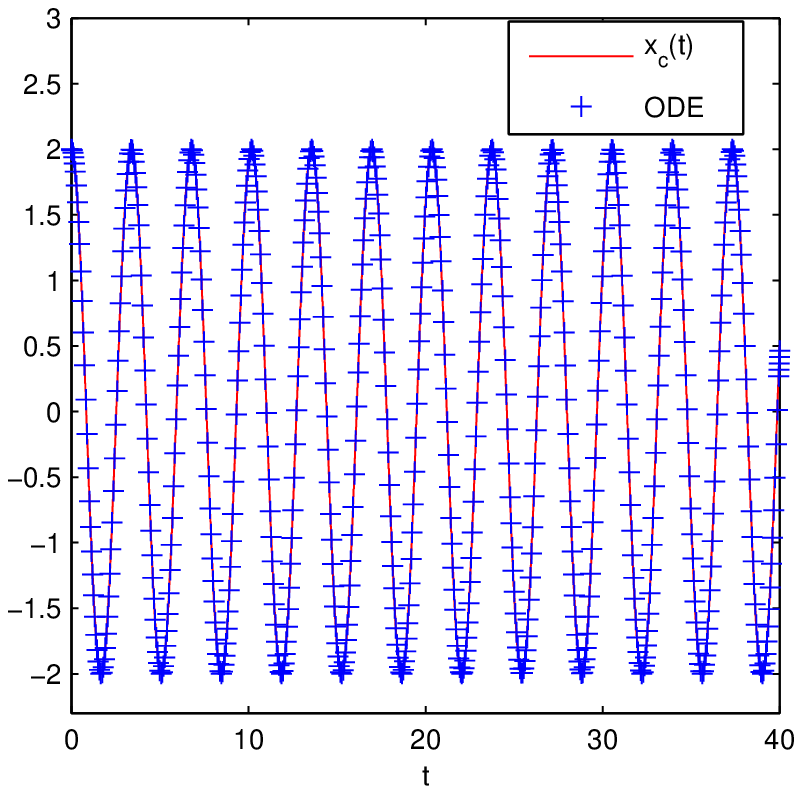,height=5cm,width=6cm,angle=0}}
\centerline{
(c)\psfig{figure=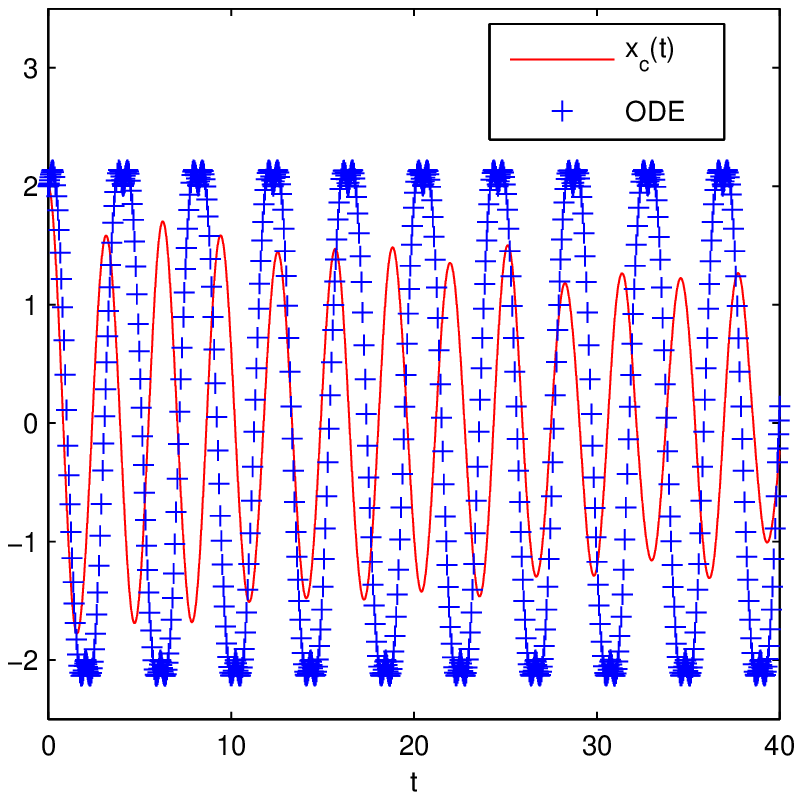,height=5cm,width=6cm,angle=0}
\quad (d)\psfig{figure=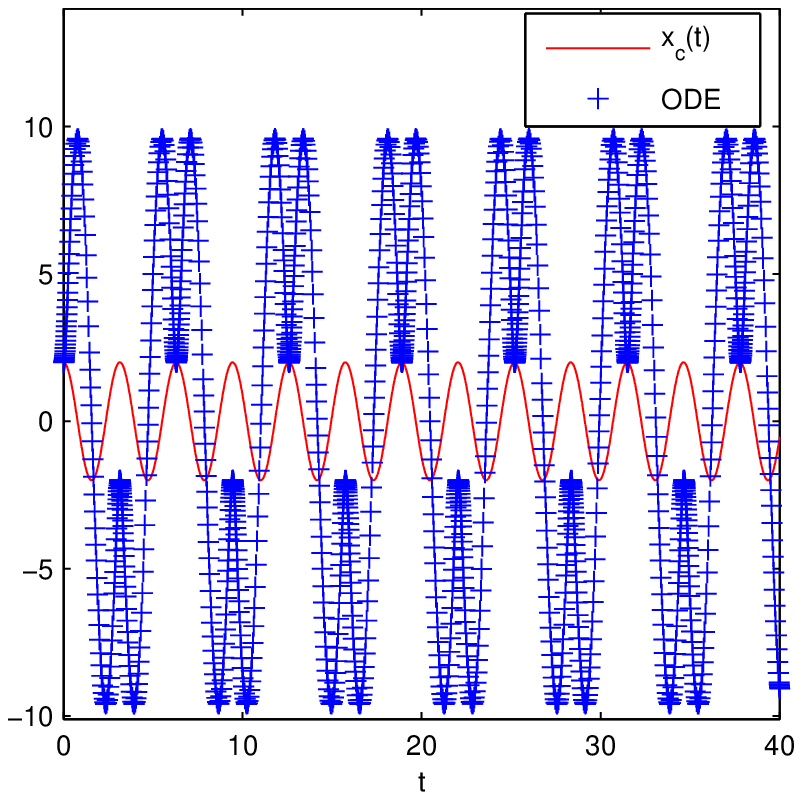,height=5cm,width=6cm,angle=0}
}

\caption{\label{fig:3}
Time evolution of the center-of-mass $x_c(t)$ for the CGPEs \eqref{eq:cgpe1}
obtained numerically from its numerical solution (i.e. labeled as $x_c(t)$ with solid lines) and
asymptotically as in Theorem  \ref{thm:com2} (i.e. labeled as 'ODE' with `+ + +')
for different sets of parameters:
(a) $(\Omega,k_0,\delta)=(50,2,0)$, (b) $(\Omega,k_0,\delta)=(50,2,10)$,
(c) $(\Omega,k_0,\delta)=(2,2,0)$,
 and (d) $(\Omega,k_0,\delta)=(50,20,0)$.   }
\end{figure}

From Figure \ref{fig:3}  and numerous tests we have done (not shown here for brevity),
we find that for the very special initial data \eqref{eq:inishift}, Theorem \ref{thm:com2}
provides a very good approximation for the dynamics of the center-of-mass over
a long time when  $|\Omega|$ is much larger than $\gamma_x$ and $k_0$.
However, when $0<\gamma_x\ll |\Omega|$ and $k_0$
is large, $x_c(t)$ behaves periodically over a long time,
but the approximation in Theorem \ref{thm:com2}  fails!
For $|\Omega|$ being comparable to $\gamma_x$, $x_c(t)$ is damped in time
and non-periodic.

\begin{remark}
 Theorem \ref{thm:com2} does not contradict with Theorem \ref{thm:comapp},
 because Theorem \ref{thm:comapp} holds for small $k_0$,
 where the $\Omega$ frequency contribution is very small  and $x_c(t)$ is almost periodic there.
 Theorem \ref{thm:com2} has certain restriction because of the assumptions we have used
 on the initial data. In particular, $k_0$ can not be large because the
 energy gap between ground modes and excited modes will be reduced for
 large $k_0$ and the assumption that the wave function remains in the ground mode will
 be violated.
\end{remark}

\subsection{Semi-classical scaling}
For strong interaction $\beta_{jl}\gg1$, we could rescale \eqref{eq:cgpe1} by choosing
$\bx\to\bx \vep^{-1/2}$ , $\psi_{j}\to\psi^\vep_j\vep^{d/4}$, $\vep=1/\beta^{2/(d+2)}$, $
\beta=\max\{|\beta_{11}|,|\beta_{12}|,|\beta_{22}|\}$, which gives the following CGPEs
\be\label{eq:cgpe1:semc}
\begin{split} &i\vep\partial_t \psi_{1}^\vep=\left[-\frac{\vep^2}{2}\nabla^2
+V_1(\bx)+ik_0\vep^{3/2}\p_x+\frac{\delta\vep}{2}
+\sum_{j=1}^2\beta_{1j}^0|\psi^\vep_{j}|^2\right]\psi^\vep_{1}+\frac{\Omega\vep}{2}
\psi^\vep_{2}, \\
&i\vep\partial _t \psi^\vep_{2}=\left[-\frac{\vep^2}{2}\nabla^2
+V_2(\bx)-ik_0 \vep^{3/2}\p_x-\frac{\delta\vep}{2}+\sum_{j=1}^2\beta_{2j}^0|\psi^\vep_{j}|^2\right]\psi^\vep_{2}
+\frac{\Omega\vep}{2}
\psi^\vep_{1},\end{split} \ee
where $\beta_{j,l}^0=\frac{\beta_{j,l}}{\beta}$ and the potential functions are given in \eqref{eq:potential:ho}. It is
of great interest to study the behavior of \eqref{eq:cgpe1:semc} when the small parameter $\vep$ tends to 0.

\textbf{Semiclassical limits in linear case}.
In the linear case, i.e. $\beta_{jl}^0=0$ for $j,l=1,2$, (\ref{eq:cgpe1:semc}) collapses to
\be\label{eq:cgpe1:semc:linear}
i\vep\p_t\Psi^\vep=\begin{bmatrix}\frac{-\vep^2}{2}\Delta+ik_0\vep^{3/2}\p_x+\frac{\delta\vep}{2}+V_1&\frac{\Omega\vep}{2}\\
\frac{\Omega\vep}{2}&\frac{-\vep^2}{2}\Delta-ik_0\vep^{3/2}\p_x-\frac{\delta\vep}{2}+V_2
\end{bmatrix}\Psi^\vep \ee
where $\Psi^\vep=(\psi_1^\vep,\psi_2^\vep)^T$.
We now describe the limit as $\vep\to0^+$ using the Wigner transform
\be
W^\vep(\Psi^\vep)(\bx,\xi)=(2\pi)^{-d}\int_{\mathbb R^d}\Psi^\vep(\bx-\vep v/2)\otimes\Psi^\vep(\bx+\vep v/2)e^{iv\cdot\xi}\,dv,
\ee
where $W^\vep$ is a $2\times2$ matrix-valued function.
The symbol corresponds to (\ref{eq:cgpe1:semc:linear}) can be written as
\be
P^\vep(\bx,\xi)=\frac{i}{2}|\xi|^2+i\begin{bmatrix}k_0\vep^{1/2}\xi_1+
V_1(\bx)+\frac{\delta\vep}{2}&\frac{\Omega\vep}{2}\\
\frac{\Omega\vep}{2}&-k_0\vep^{1/2}\xi_1+V_2(\bx)-\frac{\delta\vep}{2}\end{bmatrix},
\ee
where $\xi=(\xi_1,\xi_2,\ldots, \xi_d)^T$.
Let us consider the principal part $P$ of $P^\vep=P+O(\vep)$, i.e., we omit small term $O(\vep)$,
and we know that $-iP(\bx,\xi)$ has two eigenvalues $
\lambda_1(\bx,\xi)$ and $\lambda_2(\bx,\xi)$.
 Let $\Pi_{j}$ ($j=1,2$) be the projection matrix from $\mathbb C^2$ to the eigenvector space associated with $\lambda_{j}$. If $\lambda_{1,2}$ are
well separated, then $W^\vep(\Psi^\vep)$ converges to the
Wigner measure $W^0$ which can be decomposed as \cite{Gerard}
\be\label{eq:decwg}
W^0=u_1(\bx,\xi,t)\Pi_1+u_2(\bx,\xi,t)\Pi_2,
\ee
where $u_{j}(\bx,\xi,t)$ satisfies the Liouville equation
\be \label{Liou}
\p_tu_j(\bx,\xi,t)+\nabla_{\xi}\lambda_j(\bx,\xi,t)\cdot\nabla_{\bx}u_j(\bx,\xi,t)-
\nabla_{\bx}\lambda_j(\bx,\xi,t)\cdot\nabla_{\xi}u_j(\bx,\xi,t)=0.
\ee
It is known that such semi-classical limit fails at regions when $\lambda_1$ and $\lambda_2$ are close.

  Specifically, when $k_0=O(1)$, $\delta=O(1)$ and $\Omega=O(1)$, the limit of the
  Wigner transform $W^\vep(\Psi^\vep)$
  only has diagonal elements, and we have
\be
\label{weakwg}
 P=
\frac{i}{2}|\xi|^2+i\begin{bmatrix}
V_1(\bx)&0\\
0&V_2(\bx)\end{bmatrix}, \quad \lambda_1=\frac{1}{2}|\xi|^2+V_1(\bx), \quad \lambda_2=
\frac{1}{2}|\xi|^2+V_2(\bx).\ee
In the limit of this case, $W^0$ in \eqref{eq:decwg},  $\Pi_1$ and $\Pi_2$ are diagonal matrices,
which means the two components of $\Psi^{\vep}$ in \eqref{eq:cgpe1:semc:linear} are decoupled as $\vep\to0^+$.
In addition, the Liouville equation (\ref{Liou}) is valid with $\lambda_1$ and $\lambda_2$ defined in
(\ref{weakwg}).

Similarly, when $k_0=O(1/\vep^{1/2})$, $\delta=O(1/\vep)$ and $\Omega=O(1/\vep)$, e.g.
$k_0=\frac{k_\infty}{\vep^{1/2}}$, $\Omega=\frac{\Omega_\infty}{\vep}$ and
 $\delta=\frac{\delta_\infty}{\vep}$ with $k_\infty$, $\Omega_\infty$ and
$\delta_\infty$ nonzero constants,  the limit of the Wigner transform $W^\vep(\Psi^\vep)$
  has nonzero diagonal and off-diagonal elements, and we have
\be\label{weakwg3}
 P=\frac{i}{2}|\xi|^2+i\begin{bmatrix}k_\infty\xi_1+
V_1(\bx)+\frac{\delta_\infty}{2}&\frac{\Omega_\infty}{2}\\
\frac{\Omega_\infty}{2}&-k_\infty\xi_1+V_2(\bx)-\frac{\delta_\infty}{2}\end{bmatrix},
\ee
and
\be\label{weakwg5}
\lambda_{1,2}=\frac{|\xi|^2}{2}+\frac{V_1(\bx)+V_2(\bx)}{2}\pm\frac{\sqrt{[V_1(\bx)-
V_2(\bx)+2k_\infty\xi_1+\delta_\infty]^2+\Omega_\infty^2}}{2}.
\ee
In the limit of this case, $W^0$ in \eqref{eq:decwg}, $\Pi_1$ and $\Pi_2$ are full matrices,
which means that the two components of $\Psi^{\vep}$ in \eqref{eq:cgpe1:semc:linear}
are coupled as $\vep\to0^+$. Again, the Liouville equation (\ref{Liou}) is
valid with $\lambda_1$ and $\lambda_2$ defined in
(\ref{weakwg5}).

Of course, for the nonlinear case, i.e. $\beta_{jl}^0\ne0$ for $j,l=1,2$,
only the case when $\Omega=0$ and $k_0=0$ has been addressed \cite{Lee}.
For $\Omega\neq0$ and $k_0\neq0$, it is
still not clear about the semi-classical limit of the CGPEs
\eqref{eq:cgpe1:semc}.

\section{Conclusions}
We have studied analytically and asymptotically as well as numerically ground states and dynamics of
two-component spin-orbit-coupled Bose-Einstein condensates (BECs) based on
the coupled Gross-Pitaevskii equations (CGPEs) with
the spin-orbit (SO) and Raman couplings.
For ground state properties,  we established
existence and uniqueness, as well as non-existence of the grounds states
in different parameter regimes and studied their limiting behavior and
structure with various combination of the SO and Raman coupling strengths.
Efficient and accurate numerical methods
were designed for computing the ground states and dynamics of SO-coupled BECs,
especially for box potentials. Numerical results for the ground states were
reported under different parameter regimes, which confirmed
our analytical results on ground states.
For dynamical properties, we obtained the dynamical laws governing the motion of the
center-of-mass and showed that the dynamics of the center-of-mass in the SO-coupled direction
is either non-periodic or a periodic function with different frequency to the trapping frequency,
which is completely different from the case without SO coupling.
Numerical results were presented to confirm our asymptotical (or approximate)
results on the dynamics of the center-of-mass.
Finally, we described the semi-classical limit of the CGPEs
in the linear case via the Wigner transform method.



\begin{thebibliography}{}
\bibitem{Anderson}
{\sc M. H. Anderson, J. R. Ensher, M. R. Matthewa, C. E.
Wieman and E. A. Cornell},
\textit{Observation of Bose-Einstein
condensation in a dilute atomic vapor}, Science, 269 (1995),
pp.~198--201.

\bibitem{Antoine}
{\sc X. Antoine, W. Bao and C. Besse},
\textit{Computational methods for the dynamics of the nonlinear Schr\"odinger/Gross-Pitaevskii equations},
Comput. Phys. Commun.,  184 (2013), pp.~2621-2633.

\bibitem{Bao2014}
{\sc W. Bao},
\textit{Mathematical models and numerical methods for Bose-Einstein condensation},
Proceeding of International Congress of Mathematicians, 2014, to appear (arXiv: 1403.3884 (math.ph)).


\bibitem{Bao1}
{\sc W. Bao}, \textit{Ground states and dynamics of multicomponent
Bose--Einstein condensates}, Multiscale Model. Simul., 2 (2004),
pp.~210--236.

\bibitem{Bao2009}
{\sc W. Bao and Y. Cai},
\textit{Ground states of two-component Bose-Einstein condensates with an internal atomic Josephson junction},
East Asia J. Appl. Math., 1 (2010), pp.~49--81.

\bibitem{Bao2013}
{\sc W. Bao and Y. Cai},
\textit{Mathematical theory and numerical methods for Bose-Einstein condensation},
Kinet. Relat. Mod.,  6 (2013), pp.~1--135.

\bibitem{Bao2006}
{\sc W. Bao, I-L. Chern and F. Y. Lim},
Efficient and spectrally accurate numerical methods for computing
ground and first excited states in Bose-Einstein condensates,
J. Comput. Phys., 219 (2006), pp. 836-854.


\bibitem{Wz1}
{\sc W. Bao and Q. Du}, \textit{Computing the ground state solution
of Bose-Einstein condensates by a normalized gradient flow}, SIAM J.
Sci. Comput., 25 (2004), pp.~1674-1697.


\bibitem{Bao2003}
{\sc W. Bao, D. Jaksch and P. A. Markowich},
\textit{Numerical solution of the Gross-Pitaevskii equation for Bose-Einstein condensation},
J. Comput. Phys., 187 (2003), pp.~318 - 342.




\bibitem{Ben}
{\sc N. Ben Abdallah, F. M\'ehats, C. Schmeiser and R. M. Weish\"aupl}
\textit{The nonlinear Schr\"odinger equation with a strongly anisotropic harmonic potential},
SIAM J. Math. Anal., 47 (2005), pp. 189--199.

\bibitem{Chang}
{\sc S. M.  Chang, C. S. Lin, T. C. Lin and W. W. Lin},
\textit{Segregated nodal domains of two-dimensional multispecies
Bose-Einstein condensates}, Physica D, 196 (2004), pp.  341-361.

\bibitem{Davis}
{\sc K. B. Davis, M. O. Mewes, M. R. Andrews, N. J. van Druten,
 D. S. Durfee, D. M. Kurn and W. Ketterle},
\textit{Bose-Einstein condensation in a gas of sodium atoms}, Phys.
Rev. Lett., 75 (1995), pp. 3969--3973.

\bibitem{Deng}
{\sc Y. Deng, J. Cheng, H. Jing, C. P. Sun and S. Yi},
\textit{Spin-orbit-coupled dipolar Bose-Einstein condensates},
Phys. Rev. Lett.,  108 (2012), 125301.

\bibitem{Galitski}
{\sc V. Galitski and I. B. Spielman},
\textit{Spin-orbit coupling in quantum gases},
Nature,  494 (2012), pp.~49-54.

\bibitem{Gerard}
{\sc P. G\'erard, P. A. Markowich and NJ Mauser},
\textit{Homogenization limits and Wigner transforms},
Comm. Pure Appl. Math., 50 (1997), 323--379.

\bibitem{Hamner}
{\sc C. Hamner, Y. Zhang, M. A. Khamehchi, M. J. Davis  and P. Engels},
\textit{Spin-orbit coupled Bose-Einstein condensates in a one-dimensional optical lattice},
 arXiv:1405.4048.

\bibitem{Hasan}
{\sc M. Z. Hasan and C. L. Kane},
\textit{Colloquium: Topological insulators},
Rev. Mod. Phys., 82 (2010), 3045--3067.


\bibitem{Hu}
{\sc H. Hu, B. Ramachandhran, H. Pu and X. Liu},
\textit{Spin-orbit coupled weakly interacting Bose-Einstein condensates in harmonic traps},
Phys. Rev. Lett., 108 (2012), 010402.

\bibitem{Lee}
{\sc C.-C. Lee and T.-C. Lin},
\textit{Incompressible and compressible limits of two-component Gross¨CPitaevskii equations
with rotating fields and trap potentials},
J. Math. Phys., 49 (2008), 043517.

\bibitem{Li2012}
{\sc Y. Li, L. Pitaevskii and S. Stringari},
\textit{Quantum tricriticality and phase transitions in spin-orbit coupled Bose-Einstein
condensates}, Phys. Rev. Lett., 108 (2012), article 225301.

\bibitem{Li2013}
{\sc Y. Li, G. I. Martone, L. Pitaevskii and S. Stringari},
\textit{Superstripes and the excitation spectrum of a spin-orbit-coupled Bose-Einstein condensate},
Phys. Rev. Lett., 110 (2013), article 235302.
\bibitem{LiebL}
{\sc E. H. Lieb and M. Loss},
\textit{Analysis, Graduate Studies in Mathematics}, Amer. Math.
Soc., 2nd ed., 2001.
\bibitem{Lie}
{\sc E. H. Lieb, R. Seiringer and J. Yngvason}, \textit{Bosons in a
trap: a rigorous derivation of the Gross-Pitaevskii energy
functional}, Phy. Rev. A,  61 (2000), article 043602.

\bibitem{LS2}
{\sc E. H. Lieb and J. P. Solovej}, \textit{Ground state energy of
the two-component charged Bose gas}, Comm. Math. Phys., 252 (2004),
pp. 485-534.



\bibitem{Lin2011}
{\sc Y. J. Lin, K. Jim\'enez-Garcia and I. B. Spielman},
\textit{Spin-orbit-coupled Bose-Einstein condensates}, Nature, 471 (2011), 83--86.


 \bibitem{Liu}
 {\sc Z. Liu},
 \textit{Two-component Bose-Einstein condensates},
  J. Math. Anal. Appl., 348  (2008), pp. 274-285.

\bibitem{PitaevskiiStringari}
{\sc L. P. Pitaevskii and S. Stringari},
\textit{Bose-Einstein Condensation},
Clarendon Press, Oxford, 2003.



\bibitem{Wang}
{\sc C. Wang, C. Gao, C. Jian and H. Zhai},
\textit{Spin-orbit coupled spinor Bose-Einstein condensates},
Phys. Rev. Lett., 105 (2010), 160403.

\bibitem{WangH}
{\sc H. Wang and Z. Xu},
\textit{A projection gradient method for energy functional minimization with a
constraint and its application into computing ground state
of spin-orbit-coupled Bose-Einstein condensate},
Comp. Phys. Comm., in press.

\bibitem{Weinstein}
{\sc M. I. Weinstein}, \textit{Nonlinear Schr\"{o}dinger equations and sharp interpolation estimates}, Comm.
Math. Phys., 87 (1983), pp.~567-576.

\bibitem{Xiao}
{\sc D. Xiao, M. Chang and Q. Niu},
\textit{Berry phase effects on electronic properties},
Rev. Mod. Phys., 82 (2010), pp.~1959--2007.

\bibitem{ZhangJ}
{\sc J. Zhang, S. Ji, Z. Chen, L. Zhang, Z. Du, B. Yan, G. Pan, B. Zhao,
Y. Deng, H. Zhai, S. Chen and J. Pan}
\textit{Collective dipole oscillations of a spin-orbit coupled Bose-Einstein condensate},
Phys. Rev. Lett., 109 (2012), 115301.

\bibitem{Zhang}
{\sc Y. Zhang, L. Mao and C. Zhang},
\textit{Mean-field dynamics of spin-orbit coupled Bose-Einstein condensates},
Phys. Rev. Lett., 108 (2012), 035302.

\bibitem{Zhu}
{\sc Q. Zhu, C. Zhang and B. Wu},
\textit{Exotic superfluidity in spin-orbit coupled Bose-Einstein condensates},
 Europhys. Lett., 100 (2012), 50003.
\end{thebibliography}
\end{document}